\documentclass[a4paper,12pt]{scrartcl}
\usepackage{amsmath,amssymb,latexsym,amsthm,enumerate}
\usepackage{xcolor}
\usepackage{graphicx}
\usepackage{epstopdf}
\usepackage{changes}
\usepackage[nocompress]{cite}

\usepackage[T1]{fontenc}
\usepackage[utf8]{inputenc}
\usepackage[english]{babel}

\usepackage{hyperref}
\usepackage[capitalise,noabbrev]{cleveref}

\newcommand*{\wt}{\widetilde}
\newcommand*{\ol}{\overline}

\newcommand*{\bbN}{\mathbb N}

\newcommand*{\bbR}{\mathbb R}

\newcommand*{\cA}{\mathcal{A}}
\newcommand*{\cB}{\mathcal{B}}

\newcommand*{\cE}{\mathcal{E}}
\newcommand*{\cF}{\mathcal{F}}

\newcommand*{\N}{\mathbb{N}}

\newcommand*{\R}{\mathbb{R}}

\DeclareMathOperator*{\essinf}{ess\,inf}

\newcommand*{\eps}{\varepsilon}

\newcommand*{\sgn}{\operatorname{sgn}}

\newcommand{\be}{\begin{eqnarray*}}
\newcommand{\ee}{\end{eqnarray*}}
\newcommand{\ben}{\begin{eqnarray}}
\newcommand{\een}{\end{eqnarray}}
\newcommand{\bi}{\begin{itemize}}
\newcommand{\ei}{\end{itemize}}

\newtheorem{theo}{Theorem}[section]

\newtheorem{lemma}[theo]{Lemma}
\newtheorem{propo}[theo]{Proposition}

\theoremstyle{definition}
\newtheorem{ex}[theo]{Example}
\newtheorem{defi}[theo]{Definition}

\newtheorem{remark}[theo]{Remark}
\newtheorem{discussion}[theo]{Discussion}

\title{Self-exciting price impact\\via negative resilience in\\stochastic order books}

\author{Julia Ackermann\thanks{Institute of Mathematics, University of Gie{\ss}en, Arndtstr.~2, 35392 Gießen, Germany.
\emph{Email:} julia.ackermann@math.uni-giessen.de, \emph{Phone:} +49 (0)641 9932113.}
\and Thomas Kruse\thanks{Institute of Mathematics, University of Gie{\ss}en, Arndtstr.~2, 35392 Gießen, Germany.
\emph{Email:} thomas.kruse@math.uni-giessen.de, \emph{Phone:} +49 (0)641 9932102.}
\and Mikhail Urusov\thanks{Faculty of Mathematics, University of Duisburg-Essen, Thea-Leymann-Str.~9, 45127 Essen, Germany.
\emph{Email:} mikhail.urusov@uni-due.de, \emph{Phone:} +49 (0)201 1837428.
}}

\begin{document}

\maketitle

\begin{abstract}
Most of the existing literature on optimal trade execution in limit order book models assumes that resilience is positive.
But negative resilience also has a natural interpretation, as it models self-exciting behaviour of the price impact, where trading activities of the large investor stimulate other market participants to trade in the same direction.
In the paper we discuss several new qualitative effects on optimal trade execution that arise when we allow resilience to take negative values.
We do this in a framework where both market depth and resilience are stochastic processes.

\smallskip
\emph{Keywords:}
optimal trade execution;
limit order book;
stochastic market depth;
stochastic resilience;
negative resilience;
quadratic BSDE;
infinite-variation execution strategy;
semimartingale execution strategy.

\smallskip
\emph{2020 MSC:}
Primary: 91G10; 93E20; 60H10.
Secondary: 60G99.
\end{abstract}

\section*{Introduction}

In an illiquid financial market large orders have a substantial adverse effect on the realized prices.
It is, therefore, reasonable to divide a large order into smaller ones when an investor faces the task of closing a large position in an illiquid market.
The scientific literature on optimal trade execution problems deals with the optimization of such trading schedules.
The inputs are time horizon $T\in(0,\infty)$,
size $x\in\bbR$ (shares of a stock) of the financial position to be closed until time $T$
and model of the price impact.

The literature on optimal trade execution takes price impact as exogenously given. 
Depending on how the price impact is modeled the majority of current literature can be naturally divided into two groups.

In the first group of models, execution strategies $(X_t)_{t\in[0,T]}$ have absolutely continuous paths $t\mapsto X_t$, and the price impact at any time $t$ depends only on the derivative $\dot{X_t}$ at time $t$.
In particular, the price impact at time $t$ is independent of all orders executed at times prior to $t$ and does not influence the impact of the orders executed at times after~$t$.
Essentially, what is modeled in this approach is only market depth, and the price impact is purely instantaneous in the sense described above.\footnote{This, instantaneous, impact is alternatively called temporary impact.
There can also be a permanent component in the price impact, but it has no effect on determining optimal execution strategies.
In the literature, such models are often called the Almgren-Chriss type models; see, e.g.,
\cite{almgren2012optimal,
almgren2001optimal,
ankirchner2014bsdes}
and references therein.}

In the second group of models, trades induce a transient price impact that decays over time due to resilience effects of the price.
In such models, the execution price at time $t$ is influenced in a nontrivial way by orders filled at times prior to $t$, and the execution at time $t$ in turn influences the execution prices of subsequent orders.
Essentially, there are now two quantities to be modeled separately: market depth and resilience.
Such models are inspired by a limit order book interpretation.
The pioneering work 
Obizhaeva and Wang \cite{obizhaeva2013optimal}
models the price impact via a block-shaped limit order book (which translates into a constant market depth), where the impact decays exponentially at a constant rate.
Mathematically, it is this rate that is called \emph{resilience}.\footnote{From this perspective, the models within the first group are essentially models with infinite resilience, whereas the models in the second group are models with finite resilience.}
Our model in this paper falls into this second group.

As explained above, there is a clear qualitative difference between the models in the first and in the second group.
Moreover, this translates into qualitative differences in the optimal execution strategies.
One of the facets worth mentioning in this respect is that, as opposed to absolutely continuous strategies in the first group of models, optimal strategies in the second group are c\`adl\`ag and usually exhibit jumps\footnote{The exceptions are
Graewe and Horst \cite{graewe2017optimal} and
Horst and Xia \cite{horst2019multi},
where optimal strategies are absolutely continuous, although the models belong to the second group according to our classification.
The reason is that jumps are strongly penalized by the form of the functionals that are optimized in \cite{graewe2017optimal} and \cite{horst2019multi}.}
(in a sense, jumps at certain times allow to better exploit finite resilience).

\medskip
Most of the existing literature within the second group of models assumes that resilience is positive.
The explanation is that the impact of the trade should decay over time.
But negative resilience also has a natural interpretation, as it models self-exciting behaviour of the price impact, where trading activities of the large investor stimulate other market participants to trade in the same direction.
From this viewpoint, it seems reasonable to expect that there are (particularly unstable) periods in financial markets when the resilience is negative.
In this paper we discuss several new qualitative effects in optimal trade execution that can arise when we allow the resilience to take negative values.

In practice, resilience is difficult to estimate from real data (cf.\ Section~7.3 in Roch \cite{Roch}), and we are not aware of any empirical study of whether the resilience can be negative.
On the other hand, there recently appeared many papers on trade execution that model self-excitement of price impact in different ways, while, as explained above, negative resilience is an alternative way of modeling this effect.
As in Cay\'e and Muhle-Karbe \cite{CayeMuhleKarbe} and in Fu et al.\ \cite{FuHorstXia2020}, we motivate self-exciting price impact by the following reasons.
Imagine, for instance, a large trader performing extensive selling.
Firstly, a continued selling pressure makes it more and more difficult to find counterparties.
Secondly, such an extensive selling by the large trader may trigger stop-loss strategies by other market participants, where they start selling in anticipation of further decrease in the price.
Thirdly, extensive selling may also attract predatory traders that employ front-running strategies.
In each case, we obtain an increased price impact for subsequent trades.

For existing approaches to self-exciting price impact, see Alfonsi and Blanc \cite{AlfonsiBlanc}, Cartea et al.\ \cite{CarteaJaimungalRicci}, Cay\'e and Muhle-Karbe \cite{CayeMuhleKarbe}, Fu et al.\ \cite{FuHorstXia2020} and references therein.
We now explain that, mathematically, all these approaches and ours are pairwise substantially different.
In \cite{CayeMuhleKarbe} the framework is of the Almgren--Chriss type, and self-excitement is produced by the trades of the large trader in a way that the price impact coefficient depends on the trading activity of the large trader.
In \cite{AlfonsiBlanc} the orders of the large trader incur price impact like in the Obizhaeva--Wang model (with positive resilience), while the orders of other market participants are modeled by Hawkes processes with self-exciting jump intensities.
That is, in contrast to the previously mentioned approach, self-excitement is produced by the trades of other market participants.
\cite{CarteaJaimungalRicci} again use Hawkes processes but in a quite different way:
they consider an execution model where the large trader places limit orders whose fill rates depend on mutually exciting ``influential'' market order flows.
\cite{FuHorstXia2020} consider liquidation games between several large traders (and the corresponding mean-field limit as well as the single player subcase) with a self-exciting order flow. 
	In a sense, self-excitement in \cite{FuHorstXia2020} is ``more endogenous'' than in the other mentioned approaches (including ours, where the resilience process is exogenously given), as in \cite{FuHorstXia2020} 
	there appear ``child orders'' triggered by the large traders' trading activity, and as
	the strategies in \cite{FuHorstXia2020} come out as Nash equilibria in the game. 
	In our approach self-excitement is produced by the trades of the large trader at time instances when the resilience is negative in the Obizhaeva-Wang type model where both market depth and resilience are stochastic processes (differently from \cite{CarteaJaimungalRicci} and like in the other mentioned approaches, the large trader trades with market orders).

	Despite the differences in the set-up,  
	it is interesting to observe the following qualitative similarity in the strategies that may result from our approach and from the one in \cite{FuHorstXia2020}.
	Below we, in particular, discuss that, in our framework, it is never optimal to overshoot the execution target whenever the resilience is positive, but it can be optimal to overshoot the target if we allow the resilience to take negative values.
	In other words, in our framework, the possibility to overshoot the target is a qualitative effect of self-excitation via negative resilience.
	In the same vein, in the single player benchmark model for \cite{FuHorstXia2020} 
	without self-excitation, which goes back to Graewe and Horst \cite{graewe2017optimal}, it is not optimal to overshoot the execution target (this is observed in Theorem~2.2 of Horst and Kivman \cite{HorstKivman2021}), 
	whereas the resulting strategies in the model with self-excitation in \cite{FuHorstXia2020} do sometimes overshoot the target (cf.\ Figure~1 or Figure~2 in \cite{FuHorstXia2020}).

\bigskip
We now briefly describe our framework. 
The execution strategies are c\`adl\`ag semimartingales
$(X_t)_{t\in[0-,T]}$ with $X_{0-}=x$ and $X_T=0$.
As explained above, we need to allow for jumps (i.e., block trades).
In particular, a possibility of a block trade at time $0$ means that $X_{0-}$ can be different from $X_0$.
Notice that we do not require $X$ to have monotone paths, which means that we allow for trading in both directions within the execution strategy.
We assume that the realized price of the asset is the sum of a martingale unaffected price and a \emph{deviation process} $(D_t)_{t\in[0-,T]}$ that carries the price impact.
The inputs are the \emph{price impact process} $(\gamma_t)$ driven by \eqref{eq:dyn_gamma},
which models market depth,
and the \emph{resilience process} $(\rho_t)$.
Both enter the dynamics of the deviation process \eqref{eq:def_deviation}.
We see from \eqref{eq:def_deviation} that the sign of $(\rho_t)$ determines whether the deviation process moves back to zero or moves further away from it.
The optimal execution problem is given in~\eqref{eq:contr_prob}.

This general setting is elaborated in Ackermann et al.\ \cite{ackermann2020cadlag}, where the solution to the optimal execution problem is described via a solution to a challenging quadratic BSDE (characteristic BSDE).
In \cite{ackermann2020cadlag} it is shown that the characteristic BSDE has a solution in two specific subsettings of this general framework.
In this paper we also complement these results by establishing the existence for the characteristic BSDE in the subsetting, where the resilience process $(\rho_t)$ and the processes $(\mu_t)$ and $(\sigma_t)$ in dynamics \eqref{eq:dyn_gamma} for the price impact process $(\gamma_t)$ are independent from the driving Brownian motion $(W_t)$.
It turns out that this subsetting is feasible enough to study some new qualitative effects of negative resilience and to explicitly construct pertinent examples.

It is worth noting that the majority of papers on models with finite resilience considers execution strategies of finite variation.
In this stream of literature, strategies of infinite variation were first included by Lorenz and Schied \cite{lorenz2013drift}, where they allow for a non-martingale dynamics in the unaffected price, and hence the execution strategies need to account for the fluctuations in it.
Recently, strategies of infinite variation emerge in related frameworks of Horst and Kivman \cite{HorstKivman2021} and Fu et al.\ \cite{FuHorstXia2022}. 
In the framework of \cite{ackermann2020cadlag} we need to include strategies of infinite variation, as they actually come out as optimal trading schedules, e.g., to account for the fluctuations in $(\gamma_t)$ and $(\rho_t)$.
This comes with some adjustments in the conventional setting of the optimal execution problem, where the most important one is the term $d[\gamma,X]$ in the dynamics \eqref{eq:def_deviation} of the deviation process $(D_t)_{t\in[0-,T]}$.
As the main theme of this paper is to discuss the effects of negative resilience, we withdraw from an extended discussion of the term $d[\gamma,X]$ in  \eqref{eq:def_deviation} but rather refer an interested reader to \cite{ackermann2020cadlag}.
From this perspective, we mention that Carmona and Webster \cite{carmona2019selffinancing} provide a strong empirical evidence that trading strategies of large traders are of infinite variation nature.

We, finally, embed our paper into a broader set of related literature on optimal trade execution in models with finite resilience.
After the pioneering paper\footnote{In SSRN it appeared already in 2005.} \cite{obizhaeva2013optimal} subsequent work either extends the framework in different directions or suggests alternative frameworks with similar features.
Alfonsi et al.\ \cite{alfonsi2008constrained} study constrained portfolio liquidation in a model of the type as in \cite{obizhaeva2013optimal}.
There is a subgroup of models which include more general limit order book shapes, see
Alfonsi et al.\ \cite{alfonsi2010optimal},
Alfonsi and Schied \cite{alfonsi2010boptimal},
Predoiu et al.\ \cite{predoiu2011optimal}.
Models in another subgroup extend the exponential decay of the price impact to general decay kernels, see
Alfonsi et al.\ \cite{alfonsi2012order},
Gatheral et al.\ \cite{gatheral2012transient}.
Finite player games with deterministic model parameters and transient impact were studied by Luo and Schied \cite{LuoSchied}, Schied et al.\ \cite{SchiedStrehleZhang}, Schied and Zhang \cite{SchiedZhang} and Strehle \cite{Strehle}. 
Models with transient multiplicative price impact have recently been analyzed in Becherer et al.\ \cite{becherer2018optimala,becherer2018optimalb}, whereas Becherer et al.\ \cite{becherer2019stability} contains a stability result for the involved cost functionals.
Superreplication and optimal investment in a block-shaped limit order book model with exponential resilience is discussed in Bank and Dolinsky \cite{BD_AAP_2019,BD_Bern_2020} and in Bank and Vo\ss{} \cite{BV_SIFIN_2019}.
The present paper falls into the subgroup that studies time-dependent (possibly stochastic) market depth and resilience, see
Ackermann et al.\ \cite{ackermann2020optimal,ackermann2020cadlag},
Alfonsi and Acevedo \cite{alfonsi2014optimal},
Bank and Fruth \cite{bank2014optimal},
Fruth et al.\ \cite{fruth2014optimal,fruth2019optimal}.
To point out the difference from our present paper, we notice that all mentioned papers except \cite{ackermann2020optimal,ackermann2020cadlag} consider only positive resilience,
the framework in \cite{ackermann2020optimal} is in discrete time,
while \cite{ackermann2020cadlag} does not study the question of what kind of new effects can arise when the (time-dependent) resilience process is allowed to take negative values.

The paper is organized as follows.
\cref{sec:prob_form} contains a precise description of our setting and formulates the problem of optimal trade execution in this setting.
\cref{sec:sol_ote} describes the solution to this optimal trade execution problem based on the results from \cite{ackermann2020cadlag}.
The key ingredient here is \cref{propo:existenceBSDEindependentcoeff} that establishes existence of a solution to the characteristic BSDE in our setting.
In \cref{sec:ojpc} we present two general results about the possibility for optimal execution strategies in such models to overjump zero or to exhibit premature closure.
Loosely speaking, a necessary condition for overjumping zero or premature closure is to have negative resilience at least for some time, while a sufficient condition for that is to have negative resilience for some time close to the time horizon $T$.
See \cref{sec:ojpc} for the precise formulations and more detailed discussions.
Via case studies in \cref{sec:cs} we address several questions that arise in discussions in \cref{sec:ojpc}.
For instance, one of the examples shows that the ``close to $T\,$''-requirement in the sufficient condition mentioned above is essential.
It is worth mentioning that \cref{sec:cs} contains both examples with deterministic optimal strategies and examples with stochastic ones and, in the latter examples, the strategies are of infinite variation.
In one of other examples we see that, with resilience that can take negative values, it is possible that the optimal execution strategy closes the position at a certain point in time and reopens it immediately. Finally, the paper is concluded with a more tricky example, where the position is kept closed during a time interval, after which it is reopened again.

\section{Problem formulation}\label{sec:prob_form}

Let us introduce the stochastic order book model in which we analyze the effects of negative resilience. In \cref{rem:special_case} below we explain in which sense the model is a special case of the model considered in \cite{ackermann2020cadlag}
and in which sense not.
\cref{rem:details_fs} provides information on where to find more detailed motivations and derivations of the order book model.

We fix a terminal time $T>0$ and consider trading in the time interval $[0,T]$. 
Let $(\Omega, \cF_T, (\cF_t)_{t \in [0,T]}, P)$ be a filtered probability space that satisfies the usual conditions and supports a Brownian motion $W=(W_t)_{t\in[0,T]}$. 
Furthermore, we assume that $(\cF_t)_{t \in [0,T]}$ has the structure
$\cF_t=\bigcap_{\eps>0}(\cF^W_{t+\eps}\vee\cF^\perp_{t+\eps})$,
$t \in [0,T)$,
$\cF_T=\cF^W_T\vee\cF^\perp_T$,
where $(\cF^W_t)_{t \in [0,T]}$ denotes the filtration generated by $W$, and $(\cF^\perp_t)_{t \in [0,T]}$ is a right-continuous complete filtration such that $\cF_T^W$ and $\cF^\perp_T$ are independent.  
Throughout the paper,
$E_t[\cdot]$ denotes the conditional expectation $E[\cdot | \cF_t]$ for $t \in [0,T]$,
and $\mu_L$ denotes the Lebesgue measure on $[0,T]$.

As input processes we require three $(\cF^\perp_t)_{t\in[0,T]}$-progressively measurable processes 
$\rho=(\rho_t)_{t\in[0,T]}$, $\mu=(\mu_t)_{t\in[0,T]}$, and $\sigma=(\sigma_t)_{t\in[0,T]}$ such that there  exist deterministic $\ol c, \ol\varepsilon \in(0,\infty)$
such that
\begin{align}
2\rho_. + \mu_.-\sigma_.^2&\geq \ol \varepsilon\quad P\times\mu_L\text{-a.e.},
\label{eq:Cgeeps}\\
\max\{|\rho_.|,|\mu_.|\} &\leq \ol c\quad P\times\mu_L\text{-a.e.}
\label{eq:Cbdd}
\end{align}
Here and in what follows, we write $\rho_.$, $\mu_.$, etc., to emphasize the presence of the time variable, i.e., we do so to indicate that we speak about the process as a whole. 
Assumption \eqref{eq:Cgeeps} is a structural condition on the input processes which, roughly speaking, ensures that the minimization problem under consideration (see \eqref{eq:contr_prob} below) is convex. To see this, we refer to the alternative representation of the cost function provided in \cite[Theorem~3.1]{ackermann2020cadlag}. Note that the process $2\rho + \mu-\sigma^2$ also shows up in the denominator of the driver of the characteristic BSDE \eqref{eq:BSDEforBM}. Assumption \eqref{eq:Cbdd} is a boundedness condition that we need in order to ensure existence of a solution of BSDE \eqref{eq:BSDEforBM}. Please note that Assumption \eqref{eq:Cgeeps} in combination with Assumption \eqref{eq:Cbdd} also implies boundedness of $\sigma$. 

The two processes $\mu$ and $\sigma$ are used to model price impact. 
More precisely, we define the \emph{price impact process} $\gamma=(\gamma_t)_{t\in[0,T]}$ to be the solution of 
\begin{equation}\label{eq:dyn_gamma}
	d\gamma_t = \gamma_t \mu_t dt + \gamma_t \sigma_t dW_t, \quad t \in [0,T], \quad  \gamma_0>0,
\end{equation}
where $\gamma_0$ is a positive $\cF_0$-measurable random variable. 
Consequently, $\gamma$ is the positive continuous $(\cF_t)_{t\in[0,T]}$-adapted process 
\begin{equation*}
	\gamma_t = \gamma_0 \exp\left( \int_0^t \left( \mu_s - \frac{\sigma_s^2}{2} \right) ds + \int_0^t \sigma_s dW_s \right) ,\quad t \in [0,T].
\end{equation*}

Given an open position $x\in\R$ to be liquidated, an \emph{execution strategy} is a c\`adl\`ag semimartingale $X=(X_t)_{t\in[0-,T]}$ such that $X_{0-}=x$ and $X_T=0$. 
For any $t \in [0,T]$, the quantity $X_{t-}$ describes the remaining position to be closed during $[t,T]$. As in \cite{ackermann2020cadlag} we follow the convention that a positive position $X_{t-}>0$ means the trader has to sell an amount of $|X_{t-}|$ shares, whereas $X_{t-}<0$ requires to buy an amount of $|X_{t-}|$ shares. 
Note that we do not require an execution strategy $X$ to have monotone paths and hence we allow for selling and buying within the same strategy. Moreover, the paths of execution strategies can exhibit jumps and thus so-called block trades are possible. 

We assume that trading according to an execution strategy $X$ affects the asset price. To model this influence 
we associate to every execution strategy $X$ a \emph{deviation process} $D=(D_t)_{t \in [0-,T]}$
with initial deviation $d \in\R$. 
We assume that the actual price of the asset is the sum of an unaffected price and the price deviation $D$. 
The unaffected price is assumed to be a martingale satisfying suitable integrability assumptions. 
This ensures that the optimal trade execution problem we are about to set up (see \eqref{eq:contr_prob} below) does not depend on the unaffected price process and that we only need to focus on the deviation $D$ (see \cite[Remark 2.2]{ackermann2020cadlag} for more detail). 
The deviation is modeled as follows. 
Given $x,d \in \R$ and an execution strategy $X=(X_t)_{t\in[0-,T]}$, the deviation process $D=(D_t)_{t\in[0-,T]}$ associated to $X$ is defined by 
\begin{equation}\label{eq:def_deviation}
	dD_t = -\rho_t D_t dt + \gamma_t dX_t + d[\gamma,X]_t, \quad t \in [0,T], \quad D_{0-}=d,
\end{equation}
i.e., 
\begin{equation*}
	D_t = e^{-\int_0^t \rho_s ds } \left( d + \int_{[0,t]} e^{\int_0^s \rho_u du } \gamma_s dX_s + \int_{[0,t]} e^{\int_0^s \rho_u du } d[\gamma,X]_s \right), 
	\quad t \in [0,T], 
	\quad D_{0-}=d.
\end{equation*}
When ignoring the effects of $X$ on $D$ at time $t\in [0,T]$ we see from \eqref{eq:def_deviation} that the sign of $\rho_t$ determines whether the deviation tends back to $0$ or further moves away from it. In the case $\rho> 0$, which is typically assumed in the literature, the deviation is always reverting to $0$ and the speed of reversion is determined by the magnitude of $\rho$. The input process $\rho$ thus models how fast the order book recovers from past trades and is therefore called the \emph{resilience process}. We allow $\rho$ to also take negative values and thereby enable the incorporation of signaling effects, where, e.g., a series of buy trades might indicate the arrival of further buy trades and therefore lead to a further growth of the deviation process.

For $x,d \in \R$, we let $\cA_0(x,d)$ be the set of all execution strategies $X$ (i.e., c\`adl\`ag semimartingales $X=(X_t)_{t\in[0-,T]}$ with $X_{0-}=x$ and $X_T=0$) such that all three conditions 
\begin{equation*}
	\begin{split}
		& E_0\left[ \sup_{t\in[0,T]} \gamma_t^2 (X_t - \gamma_t^{-1} D_t)^4 \right] < \infty \text{ a.s.}, \\
		& E_0\left[ \left( \int_0^T \gamma_t^2 (X_t - \gamma_t^{-1} D_t)^4 \sigma_t^2 dt \right)^{\frac12} \right] < \infty \text{ a.s.}, \\
		& E_0\left[ \left( \int_0^T D_t^4 \gamma_t^{-2} \sigma_t^2 dt \right)^{\frac12} \right] < \infty \text{ a.s.}
	\end{split}
\end{equation*}
are satisfied.

Given $x,d \in \R$ and $X \in \cA_0(x,d)$, we then consider the expected costs 
\begin{equation}\label{eq:ex_cost}
	J(x,d,X)=E_0\left[ \int_{[0,T]} D_{t-} dX_t + \int_{[0,T]} \frac{\gamma_t}{2} d[X]_t \right].
\end{equation}
The optimal trade execution problem considered here consists in minimizing the expected costs over $X \in \cA_0(x,d)$. 
An optimal strategy is an execution strategy $X^* \in \cA_0(x,d)$ such that 
\begin{equation}\label{eq:contr_prob}
J(x,d,X^*)=\essinf_{X \in \cA_0(x,d)}J(x,d,X).
\end{equation}

We point out that possible jumps of the integrators at time $0$ contribute to the integrals $\int_{[0,t]} \ldots dX_s$, $\int_{[0,t]} \ldots d[X]_s$, and $\int_{[0,t]} \ldots d[\gamma,X]_s$ in the definition of the deviation~\eqref{eq:def_deviation} and the expected costs~\eqref{eq:ex_cost}. 

\begin{remark}\label{rem:details_fs}
We refer to the introduction of \cite{ackermann2020cadlag} as well as Sections~4,~5 and Appendix~A therein for a discussion of the specific form of the deviation dynamics \eqref{eq:def_deviation} and the expected costs \eqref{eq:ex_cost}. 
In short, they come from a block-shaped symmetric limit order book model, and, e.g., the term $d[\gamma,X]$ in \eqref{eq:def_deviation} appears because execution strategies are not necessarily of finite variation. 
We also mention \cite{carmona2019selffinancing} for empirical evidence that, in related settings, trading strategies are of infinite variation nature.
\end{remark}

\begin{remark}\label{rem:special_case}
Let us briefly explain in which sense the setting outlined above is a special case of the model considered in \cite{ackermann2020cadlag}. In \cite{ackermann2020cadlag} we consider a general continuous local martingale $M$ instead of the Brownian motion $W$ to drive the price impact process $\gamma$ in \eqref{eq:dyn_gamma}. Moreover, in \cite{ackermann2020cadlag} we do not require that the three input processes $\rho$, $\mu$, and $\sigma$ are independent of the martingale $M$.

In \cite[Section 7]{ackermann2020cadlag} we establish existence of a solution to a characteristic BSDE (see \cite[(3.2) \& (3.3)]{ackermann2020cadlag}) for the trade execution problem in two subsettings: The first one assumes that $\sigma \equiv 0$, whereas the second one assumes that the underlying filtration is continuous (in the sense that every martingale is continuous). In the present paper we complement these results by establishing in \cref{propo:existenceBSDEindependentcoeff} below existence of a solution to the BSDE in the subsetting outlined above. This is included in neither of the two subsettings considered in \cite{ackermann2020cadlag}. 
\end{remark}

\section{Solution of the trade execution problem}\label{sec:sol_ote}
In this section we provide a probabilistic solution of the optimal trade execution problem \eqref{eq:contr_prob}. To this end, we establish in \cref{sec:bsde} an existence result for a characteristic BSDE associated to problem \eqref{eq:contr_prob}. Subsequently, in \cref{sec:opt_strat}, we combine this result with the main results in \cite{ackermann2020cadlag} to obtain a representation of the optimal strategy for \eqref{eq:contr_prob}.

\subsection{Characteristic BSDE}\label{sec:bsde}

A representation of the minimal expected costs, a characterization for the existence of an optimal strategy as well as a formula for the optimal strategy in the general framework considered in \cite{ackermann2020cadlag} are provided as a main result in \cite[Theorem 3.4]{ackermann2020cadlag}. This result is based on the existence of a solution of a certain BSDE~\cite[(3.2) \& (3.3)]{ackermann2020cadlag}. 

In the setting of the present paper, where $\rho,\mu,\sigma$ are independent of the Brownian motion $W$, we consider the BSDE
\begin{equation}\label{eq:BSDEforBM}
	dY_t=-\left[-\frac{
		(\rho_t + \mu_t)^2Y_t^2}{\sigma_t^2Y_t+\frac{1}{2}\left(2\rho_t +\mu_t-\sigma^2_t\right)}+\mu_t Y_t\right] dt+dM^\perp_t,
	\quad t \in [0,T], 
	\quad  Y_{T}=\frac{1}{2}.
\end{equation}
By a \emph{solution of~\eqref{eq:BSDEforBM}} we mean a pair $(Y,M^\perp)$ such that
\begin{itemize}
\item
\eqref{eq:BSDEforBM} is satisfied $P$-a.s.,
\item
$Y$ is an $(\cF_t)_{t\in[0,T]}$-adapted, c\`adl\`ag, $[0,1/2]$-valued process, and
\item
$M^\perp$ is a c\`adl\`ag $(\cF_t)_{t\in [0,T]}$-martingale with $M_0^\perp=0$, $E([M^\perp]_T)<\infty$, and $[M^\perp,W]=0$.
\end{itemize}

We show in \Cref{propo:existenceBSDEindependentcoeff} the existence\footnote{Uniqueness of the solution will follow as a byproduct of our analysis; see \cref{theo:summary} below.} of a solution $(Y,M^\perp)$ of~\eqref{eq:BSDEforBM}.
Any such solution $(Y,M^\perp)$ of~\eqref{eq:BSDEforBM} provides a solution $(Y,0,M^\perp)$ of the BSDE~\cite[(3.2) \& (3.3)]{ackermann2020cadlag} (in the sense that \cite[(3.4)]{ackermann2020cadlag} is satisfied). In particular,  we can invoke in \cref{sec:opt_strat} below the main results from \cite[Section~3]{ackermann2020cadlag}.

\begin{theo}\label{propo:existenceBSDEindependentcoeff}
	Under \eqref{eq:Cgeeps} and \eqref{eq:Cbdd} there exists a solution $(Y,M^\perp)$ of~\eqref{eq:BSDEforBM}.
\end{theo}

\begin{proof}
	Let $L\colon \R \to [0,1/2]$ be the truncation function defined by $L(y)=(y\vee 0)\wedge \frac{1}{2}$, $y\in \R$.
	Let $\overline f\colon \Omega \times [0,T]\times \R \to \R$ be the function defined by
	\begin{equation}\label{eq:trunc_driver}
		\overline f(t,y)=-\frac{
			(\rho_t + \mu_t)^2L(y)^2}{\sigma_t^2L(y)+\frac{1}{2}\left(2\rho_t +\mu_t-\sigma^2_t\right)}+\mu_t L(y), \quad t \in [0,T], y \in\R.
	\end{equation}
	We first consider BSDE~\eqref{eq:BSDEforBM} with its driver replaced by $\overline f$ and on the filtered probability space $(\Omega, \cF^\perp_T, (\cF^\perp_t)_{t \in [0,T]}, P|_{\cF_T^\perp})$, where $P|_{\cF_T^\perp}$ denotes the probability measure $P$ restricted to the sigma algebra $\cF_T^\perp$.
Note that the expressions ``$P$-a.s.'' and ``$P|_{\cF_T^\perp}$-a.s.'' have the same meaning.
In the calculations below we assume without loss of generality that $\rho$, $\mu$, and $\sigma$ satisfy \eqref{eq:Cgeeps} and \eqref{eq:Cbdd} \emph{for all} $(\omega,t)$, as we can otherwise replace them in $\ol f$ with $(\cF^\perp_t)$-progressively measurable processes $\ol\rho$, $\ol\mu$, and $\ol\sigma$ that satisfy \eqref{eq:Cgeeps} and \eqref{eq:Cbdd} \emph{for all} $(\omega,t)$ and such that
$\ol\rho=\rho$ $P\times\mu_L$-a.e.,
$\ol\mu=\mu$ $P\times\mu_L$-a.e., and
$\ol\sigma=\sigma$ $P\times\mu_L$-a.e.
Observe that for all $t\in [0,T]$ the function $\R \ni y \mapsto \overline f(t,y)$ is continuous. Moreover, it
	is concave on $[0,1/2]$ and constant on the complement $\R \setminus [0,1/2]$. 
	This implies for all $t\in [0,T]$ and $y'<y$ that 
	\begin{equation*}
		\frac{\overline f(t,y)- \overline f(t,y')}{y-y'} 
		\leq \max\left\{ \frac{\partial^+ \overline f}{\partial y}(t,0), 0 \right\} = \max\left\{ \mu_t, 0 \right\} \le \ol c .
	\end{equation*}
	It follows that for all $t\in [0,T]$ and all $y,y' \in \R$ it holds 
	\begin{equation*}
		\big(\overline f(t,y)- \overline f(t,y') \big) (y-y') \le \ol c (y-y')^2.
	\end{equation*}
	Moreover, it holds for all $t \in [0,T]$ that 
	\begin{equation*}
		\begin{split}
			\sup_{y\in \R}\lvert\overline f(t,y) - \overline f(t,0)\rvert
			&\le \sup_{y\in \R} 
			\left\lvert \frac{
				(\rho_t + \mu_t)^2L(y)^2}{\sigma_t^2L(y)+\frac{1}{2}\left(2\rho_t +\mu_t-\sigma^2_t\right)}\right\rvert
				+\sup_{y\in \R} \lvert\mu_t L(y)\rvert\\
			& \leq \frac{2\ol c^2}{\ol \varepsilon} + \frac12 \ol c .
		\end{split}
	\end{equation*}
	This implies in particular that $\sup_{y\in \R}\lvert\overline f(\cdot,y) - \overline f(\cdot,0)\rvert \in L^2(\Omega \times [0,T])$. 
	By \cite[Proposition 5.1]{klimsiak2021nonlinear} (see also \cite[Theorem 1]{kruse2016bsdes}) 
	there exists a pair $(Y,M^\perp)$ 
	such that BSDE~\eqref{eq:BSDEforBM} on $(\Omega, \cF^\perp_T, (\cF^\perp_t)_{t \in [0,T]}, P|_{\cF_T^\perp})$ with its driver replaced by $\overline f$ is satisfied a.s., $Y$ is a c\`adl\`ag $(\cF_t^\perp)_{t\in[0,T]}$-adapted process with $E[\sup_{t \in [0,T]} Y_t^2]<\infty$, and 
	$M^\perp$ is a c\`adl\`ag $(\cF_t^\perp)_{t\in [0,T]}$-martingale with $M_0^\perp=0$ and  $E[[M^\perp]_T]<\infty$.
	
	Next, note that $(0,0)$ is a solution of the BSDE with driver $\overline f$ and terminal condition $0$. Then a comparison principle for BSDEs (e.g., \cite[Proposition 4]{kruse2016bsdes}) proves that $Y_t\ge 0$ for all $t\in [0,T]$ a.s. Furthermore, it holds a.s.\ that for all $t\in [0,T]$
	\begin{equation*}
		\overline f(t,1/2)=-\frac{
			(\rho_t + \mu_t)^2}{2\left(2\rho_t +\mu_t\right)}+\frac{1}{2}\mu_t=-\frac{\rho_t^2}{2(2\rho_t +\mu_t)}\le 0 .
	\end{equation*}
	Again the comparison principle ensures that $Y_t\le \frac{1}{2}$ for all $t\in [0,T]$ a.s. 
	In particular, the truncation function in \eqref{eq:trunc_driver} is inactive. 
	Therefore, $(Y,M^\perp)$ also a.s.\ satisfies BSDE~\eqref{eq:BSDEforBM} and $Y$ is $[0,1/2]$-valued.

Since $\cF_T^W$ and $\cF_T^\perp$  are independent and
$\cF_t=\bigcap_{\eps>0}(\cF_{t+\eps}^W\vee \cF_{t+\eps}^\perp)$ for all $t\in[0,T)$,
we have that $M^\perp$ is not only an $(\cF_t^\perp)_{t\in [0,T]}$-martingale, but also an $(\cF_t)_{t\in [0,T]}$-martingale. 
	Furthermore, we can show that $M^\perp W$ is an $(\cF_t)_{t\in [0,T]}$-martingale. 
	It follows that $\langle M^\perp, W \rangle =0$. Since $W$ is continuous, it holds that $[M^\perp,W]$ is continuous, and hence $[M^\perp,W]=0$.
	This completes the proof.
\end{proof}

\subsection{Representation of the optimal strategy}\label{sec:opt_strat}

For a solution $(Y,M^\perp)$ of \eqref{eq:BSDEforBM}, we recall from \cite[(3.5)]{ackermann2020cadlag} the 
process $\wt\beta=(\wt\beta_t)_{t \in [0,T]}$ defined by, in the present set-up, 
\begin{equation}\label{eq:wtbetawoZ}
	\wt\beta_t = \frac{(\rho_t + \mu_t)Y_t}{\sigma_t^2 Y_t + \frac12 (2\rho_t + \mu_t - \sigma_t^2 )}, \quad t \in[0,T].
\end{equation}
By \eqref{eq:Cgeeps}, \eqref{eq:Cbdd}, and the fact that $Y$ is $[0,1/2]$-valued, 
we have that $\wt\beta$ is $P\times\mu_L$-a.e.\ bounded.
It thus follows from \cite[Proposition 3.8]{ackermann2020cadlag} that the solution of \eqref{eq:BSDEforBM} is unique
up to indistinguishability.
Furthermore, by boundedness of $\wt\beta$, under the condition that 
\begin{align}
	&\exists \text{ a c\`adl\`ag semimartingale } \beta \text{ such that } \wt\beta=\beta \,\, P\times\mu_L\text{-a.e.}\label{eq:exbeta}
\end{align}
we obtain from \cite[Theorem 3.4]{ackermann2020cadlag} for any initial values $x,d\in\R$ (see also \cite[Lemma~3.3]{ackermann2020cadlag} for the case $x=\frac{d}{\gamma_0}$) the existence of an optimal strategy,
which is unique up to $P\times\mu_L$-null sets.
Notice that, in our present context, this is equivalent to uniqueness up to indistinguishability.
Indeed, if $X^{(1)}$ and $X^{(2)}$ are optimal strategies, then they are indistinguishable, as $X^{(1)}$ and $X^{(2)}$ are c\`adl\`ag and $X^{(1)}=X^{(2)}$ $P\times\mu_L$-a.e.
Note that condition \eqref{eq:exbeta} is in particular guaranteed if $\rho,\mu,\sigma$ are deterministic and of finite variation, as in the examples in \Cref{sec:cs} below.

\bigskip
We extract the following representation for the optimal strategy and its associated deviation from \cite[Theorem 3.4]{ackermann2020cadlag} provided that \eqref{eq:exbeta} holds true.
Define
\begin{equation*}
	Q_t = - \int_0^t \beta_s \sigma_s dW_s - \int_0^t \beta_s (\mu_s + \rho_s - \sigma_s^2) ds, \quad t\in [0,T],
\end{equation*}
and denote by $\cE(Q)$ the stochastic exponential of $Q$, i.e., 
\begin{equation*}
	\begin{split}
		\cE(Q)_t & = \exp\left\{ - \int_0^t \Big(\beta_s (\mu_s + \rho_s - \sigma_s^2) + \frac12 \beta_s^2 \sigma_s^2\Big)\,ds 
		- \int_0^t \beta_s \sigma_s\,dW_s \right\}, \quad t \in [0,T].
	\end{split}
\end{equation*}
Let $x,d \in \R$. 
Then the optimal strategy
$\left( X^*_t \right)_{t \in [0-,T]} \in \cA_0(x,d)$ 
is given by the formulas
\begin{equation}\label{eq:optstrat}
	\begin{split}
		& X^*_{0-}=x, \quad X^*_T=0,\\
		& X^*_t = 	\left(x-\frac{d}{\gamma_0}\right)(1-\beta_t) \cE(Q)_{t},
		\quad t\in [0,T). 
	\end{split}
\end{equation}
The associated deviation process 
$(D^*_t)_{t\in[0-,T]}$ is given by 
\begin{equation}\label{eq:optdev}
	\begin{split}
		& D^*_{0-}=d, \quad D^*_T=\left(x-\frac{d}{\gamma_0}\right)(-\gamma_T)\cE(Q)_{T},\\
		& D^*_t = 	\left(x-\frac{d}{\gamma_0}\right)(-\gamma_t\beta_t)\cE(Q)_{t},
		\quad t\in [0,T). 
	\end{split}
\end{equation}

We summarize the statements above in the following theorem.

\begin{theo}\label{theo:summary}
Assume that \eqref{eq:Cgeeps} and \eqref{eq:Cbdd} hold true. Then the solution of the BSDE \eqref{eq:BSDEforBM} is unique
up to indistinguishability.
If, in addition, \eqref{eq:exbeta} is satisfied, then for all $x,d \in \R$ the unique (up to indistinguishability)
optimal strategy $X^*$ for \eqref{eq:contr_prob} is given by \eqref{eq:optstrat} and the associated deviation process $D^*$ satisfies \eqref{eq:optdev}.
\end{theo}

\section{Overjumping zero and premature closure}\label{sec:ojpc}

In this section we study qualitative effects of negative resilience on the optimal strategy. 
In particular, we examine effects that we call \emph{overjumping zero} and \emph{premature closure}. Roughly speaking, we are interested in market situations where it is optimal to change a buy program into a sell program (or vice versa), or where it is optimal to close the position strictly before the end of the execution period.
More precisely, we intend to identify market conditions under which
paths of optimal trade execution strategies with positive probability jump over the target level $0$ or already take the value $0$ prior to $T$. 
To this end recall that under \eqref{eq:Cgeeps}, \eqref{eq:Cbdd}, and \eqref{eq:exbeta}, given an initial position $x\in \R$ and an initial deviation $d\in \R$, the optimal strategy $X^*$ satisfies for all $t\in [0,T)$ that 
\begin{equation}
X^*_{0-}=x \quad \text{and} \quad X^*_t = 	\left(x-\frac{d}{\gamma_0}\right)(1-\beta_t) \cE(Q)_{t}.
\end{equation}
This representation allows to disentangle the contributions to the optimal strategy's sign of the initial conditions $x$ and $d$ on the one side and the input processes $\rho$, $\mu$, and $\sigma$ defining the market dynamics on the other side. 
Indeed, 
since the stochastic exponential $\cE(Q)$ is positive, the sign of $X^*_t$ for $t\in [0,T)$ is determined by the signs of the two factors $(x-\frac{d}{\gamma_0})$ and $(1-\beta_t)$. The first factor $(x-\frac{d}{\gamma_0})$ is determined by the initial conditions, does not depend on time, and thus can only contribute to a change of sign of $X^*$ at time $0$.
Note that $(x-\frac{d}{\gamma_0})$ has a different sign than the initial condition $X^*_{0-}=x$ if and only if $\gamma_0|x|<\sgn(x)d$. A nonzero initial deviation $d\neq 0$ can thus have the effect that $X^*$ changes its sign directly at time $0$. In practice, one would typically assume that $d=0$, in which case this factor does not contribute to a change of sign.

In the sequel we focus on the contribution of the second factor $(1-\beta)$ 
and provide definitions of the effects \emph{overjumping zero} and \emph{premature closure} which are only built upon $(1-\beta)$.
This factor and hence also these effects are determined by the input processes $\rho$, $\mu$, and $\sigma$ driving the market dynamics
and are independent of the initial conditions $x$ and $d$.

For ease of notation, we extend the domain of $\beta$ to the point $0-$ by setting $\beta_{0-}=0$.
In what follows, we denote by $\pi_\Omega$ the projection operator from $\Omega\times[0,T]$ onto $\Omega$.

\begin{defi}\label{def:131121a1}
Assume that \eqref{eq:Cgeeps}, \eqref{eq:Cbdd}, and \eqref{eq:exbeta} hold true.
Define
\begin{align*}
A_{oj}&=\{(\omega,t)\in\Omega\times[0,T):(1-\beta_{t-}(\omega))(1-\beta_t(\omega))<0\},\\
A_{pc}&=\{(\omega,t)\in\Omega\times[0,T):(1-\beta_{t-}(\omega))(1-\beta_t(\omega))=0\}.
\end{align*}

(i) We say that  
\emph{overjumping zero}
is optimal in the limit order book model driven by $\rho$, $\mu$, and $\sigma$, if $P(\pi_\Omega(A_{oj}))>0$.

\smallskip
(ii) We say that
\emph{premature closure}
is optimal, if $P(\pi_\Omega(A_{pc}))>0$.
\end{defi}

In relation with Definition~\ref{def:131121a1} we need to make the following comments.
\begin{enumerate}[(a)]
\item
$\pi_\Omega(A_{oj}),\pi_\Omega(A_{pc})\in\cF_T$
by the measurable projection theorem \cite[Theorem I.4.14]{revuzyor}
(recall that $\cF_T$ is complete and notice that
$A_{oj},A_{pc}\in\cF_T\otimes\cB([0,T])$
and, moreover, are optional sets, as $\beta$ is adapted and c\`adl\`ag).
\item
The terms \emph{overjumping zero} and \emph{premature closure} are well-defined, as $\beta$ satisfying \eqref{eq:exbeta} is unique up to indistinguishability.
\end{enumerate}

It is worth noting that the terms \emph{overjumping zero} and \emph{premature closure} could be equivalently defined with the help of stopping times:

\begin{lemma}\label{lem:131121a1}
Assume that \eqref{eq:Cgeeps}, \eqref{eq:Cbdd}, and \eqref{eq:exbeta} hold true. Then,
overjumping zero (resp., premature closure) is optimal if and only if there exists a stopping time $\tau\colon\Omega\to[0,T]$ such that $P(\tau<T)>0$ and
$$
(1-\beta_{\tau-})(1-\beta_\tau)<0
\;\;(\text{resp., }=0)\;\;
P\text{-a.s. on }\{\tau<T\}.
$$
\end{lemma}

\cref{lem:131121a1} easily follows from the optional section theorem \cite[Theorem IV.5.5]{revuzyor}, which applies because $A_{oj}$ and $A_{pc}$ are optional sets. We also remark that a simple attempt to define $\tau$ as, say, $T\wedge\inf\{t\in[0,T):(1-\beta_{t-})(1-\beta_t)<0\}$ does not always work, as, for $\omega$ such that $\tau<T$ but the infimum is not attained,
the expression $(1-\beta_{\tau-})(1-\beta_\tau)$ will be zero.

\bigskip
We now turn to the question about new qualitative effects we can get if we allow for negative resilience.
Informally, with positive resilience one will not be able to observe overjumping zero or premature closure in the optimal strategy.
On the contrary, if we allow the resilience to take negative values, then overjumping zero and premature closure in the optimal strategy become possible.
\Cref{lem:positiveresilience} and \Cref{lem:negresclosetoterminaltime}
contain precise mathematical formulations of these statements.
After these propositions we also provide a more detailed informal \cref{disc:131121a1}.

\begin{propo}\label{lem:positiveresilience}
	Assume \eqref{eq:Cgeeps} and \eqref{eq:Cbdd}. 
	
	\smallskip
	(i) We have
	$$
	\wt \beta_.\le \left(1-\frac{\rho_.}{2\rho_.+\mu_.}\right)\mathbf{1}_{\{\rho_.+\mu_.>0\}}\le 1\;\;
	P\times\mu_L\text{-a.e.\ on }
	\{(\omega,t)\in\Omega\times[0,T]:\rho_t(\omega)\ge0\}.
	$$
	
	(ii) Assume \eqref{eq:exbeta} and that $\rho\ge 0$ $P\times\mu_L$-a.e.
	Then overjumping zero is not optimal.
	
	\smallskip
	(iii) Assume \eqref{eq:exbeta} and that there exists an $\cF_T$-measurable random variable $\delta$ such that
	\begin{equation}\label{eq:131121a1}
	\delta>0\;\;P\text{-a.s.\ and }\rho_.\ge\delta\;\;P\times\mu_L\text{-a.e.}
	\end{equation}
	Then neither overjumping zero nor premature closure is optimal.
\end{propo}

In relation with \cref{lem:positiveresilience} we make the following comments.
\begin{enumerate}[(a)]
\item
A rather widespread situation in today's literature on resilient price impact is to assume a constant resilience.
This falls into part~(iii) of \cref{lem:positiveresilience}.
To discuss the assumption in~(iii) in more detail, we remark that, if
\begin{equation}\label{eq:131121a2}
\inf_{t\in[0,T]}\rho_t>0\;\;P\text{-a.s.,}
\end{equation}
then \eqref{eq:131121a1} is satisfied.
Indeed, in this case we can take $\delta=\inf_{t\in[0,T]}\rho_t$ because, by the measurable projection theorem, for all $z\in\bbR$ we have
$$
\{\omega\in\Omega:\inf_{t\in[0,T]}\rho_t<z\}=
\pi_\Omega(\{(\omega,t)\in\Omega\times[0,T]:\rho_t(\omega)<z\})\in\cF_T,
$$
i.e., $\delta:=\inf_{t\in[0,T]}\rho_t$ is $\cF_T$-measurable. More precisely, \eqref{eq:131121a1} is slightly weaker than \eqref{eq:131121a2} and can be, in fact, equivalently expressed as follows: there exists an $\cF_T\otimes\cB([0,T])$-measurable $\wt\rho$ such that $\wt\rho=\rho$ $P\times\mu_L$-a.e.\ and $\inf_{t\in[0,T]}\wt\rho_t>0$ $P$-a.s.

\item
The observation in part~(iii) of \cref{lem:positiveresilience} is in line with \cite{HorstKivman2021}, where in a different but related setting (with a positive stochastically varying resilience) it is observed that the optimal strategy never changes its sign (see \cite[Theorem~2.2]{HorstKivman2021}), which means in our terminology that neither overjumping zero nor premature closure is optimal.

\item
Comparison of (ii) and (iii) poses the question if premature closure can be optimal with nonnegative resilience.
The answer is affirmative:
e.g., if $\rho \equiv 0$, then $\beta_t =1$ for all $t \in [0,T]$,
and the optimal strategy is to close the position immediately
(cf.\ \cite[Proposition 3.7]{ackermann2020cadlag}).
This is, however, a rather degenerate example.
A much more interesting one,
for which we, however, allow the resilience to be negative,
is presented in \cref{sec:premclosure}.
\end{enumerate}

\begin{proof}[Proof of \cref{lem:positiveresilience}]
	(i) Define
	\begin{align*}
	B=\{(\omega,t)\in\Omega\times[0,T]:\quad&Y_t(\omega)\in[0,1/2],\\
	&2\rho_t(\omega)+\mu_t(\omega)-\sigma_t^2(\omega)>0,\\
	&\rho_t(\omega)\ge0\}
	\end{align*}
	and observe that $B\in\cF_T\otimes\cB([0,T])$.
	It is enough to show the claim for every $(\omega,t)\in B$.
	To this end, we fix an arbitrary $(\omega,t)\in B$.
	By \eqref{eq:wtbetawoZ} we have to show that
	\begin{equation}\label{eq:prop32_toshow}
	\frac{(\rho_t(\omega) + \mu_t(\omega))Y_t(\omega)}{\sigma_t^2(\omega) Y_t(\omega) + \frac12 (2\rho_t(\omega) + \mu_t(\omega) - \sigma_t^2(\omega) )}\le \left(\frac{\rho_t(\omega)+\mu_t(\omega)}{2\rho_t(\omega)+\mu_t(\omega)}\right)
	\mathbf{1}_{\{\rho_t(\omega)+\mu_t(\omega)>0\}}.
\end{equation}	
	If $\rho_t(\omega)+\mu_t(\omega)\le 0$, this inequality is evident. Therefore we assume $\rho_t(\omega)+\mu_t(\omega)>0$ in the sequel. Note that the fact that $Y_t(\omega)\le \frac{1}{2}$ implies
	\begin{equation*}
	\begin{split}
	Y_t(\omega)-\frac{1}{2}\frac{2\rho_t(\omega) + \mu_t(\omega) - \sigma_t^2(\omega) }{2\rho_t(\omega)+\mu_t(\omega)}
	&\le 
	Y_t(\omega)\left(1-\frac{2\rho_t(\omega) + \mu_t(\omega) - \sigma_t^2(\omega) }{2\rho_t(\omega)+\mu_t(\omega)}\right)\\
	&=
	\frac{\sigma_t^2(\omega)Y_t(\omega)}{2\rho_t(\omega)+\mu_t(\omega)}.
\end{split}	
	\end{equation*}
This shows that
	\begin{equation*}
	\begin{split}
	(\rho_t(\omega)+\mu_t(\omega))Y_t(\omega)\le 
\frac{\rho_t(\omega)+\mu_t(\omega)}{2\rho_t(\omega)+\mu_t(\omega)}	
\left(
\sigma_t^2(\omega)Y_t(\omega)+\frac{1}{2}(2\rho_t(\omega) + \mu_t(\omega) - \sigma_t^2(\omega))
\right)
\end{split}	
	\end{equation*}
	and hence establishes \eqref{eq:prop32_toshow}.
	
	\smallskip
	(ii) We first notice that (i) and \eqref{eq:exbeta} ensure that
	$\beta_.\le 1$ $P\times\mu_L$-a.e.
	As $\beta$ has c\`adl\`ag paths, by the standard Fubini argument, we infer that $P$-a.s.\ it holds:
	for all $t\in [0,T]$, we have $\beta_t\le 1$.
	This shows that overjumping zero is not optimal.
	
	\smallskip
	(iii) It suffices to show that premature closure is not optimal. Define
	\begin{align*}
	C=\{(\omega,t)\in\Omega\times[0,T]:\quad&Y_t(\omega)\in[0,1/2],\\
	&2\rho_t(\omega)+\mu_t(\omega)-\sigma_t^2(\omega)>0,\\
	&\max\{|\rho_t(\omega)|,|\mu_t(\omega)|\}\le\ol c,\\
	&\delta(\omega)>0\text{ and }\rho_t(\omega)\ge\delta(\omega)\},
	\end{align*}
	where $\ol c$ is from \eqref{eq:Cbdd}, and notice that $C\in\cF_T\otimes\cB([0,T])$.
	It follows from (i) that
	$$
	\wt\beta_t(\omega)\le\max\left\{1-\frac{\delta(\omega)}{3\ol c},0\right\}<1\quad\text{for all }(\omega,t)\in C.
	$$
	As $P\times\mu_L((\Omega\times[0,T])\setminus C)=0$ and $\beta$ is c\`adl\`ag, we conclude that $P$-a.s.\ it holds
	$$
	\sup_{t\in [0,T]} \beta_t\le\max\left\{1-\frac{\delta}{3\ol c},0\right\}<1
	$$
	(again by the Fubini argument), and hence premature closure is not optimal.
\end{proof}

In the sequel, for a set $K\subseteq\Omega\times[0,T]$ and $\omega\in\Omega$, we use the notation
$$
K_\omega=\{t\in[0,T]:(\omega,t)\in K\}
$$
for the section of $K$.
We will permanently use the well-known statements that,
if $K\in\cF_T\otimes\cB([0,T])$, then
\begin{itemize}
\item
for any $\omega\in\Omega$, $K_\omega\in\cB([0,T])$,
\item
and the mapping $\omega\mapsto\mu_L(K_\omega)$ is $\cF_T$-measurable.
\end{itemize}

\begin{propo}\label{lem:negresclosetoterminaltime}
	Assume \eqref{eq:Cgeeps}, \eqref{eq:Cbdd}, and \eqref{eq:exbeta}.
	In addition, assume that there exists an $\cF_T$-measurable random variable $\delta$ such that
	\begin{equation}\label{eq:131121b1}
	P\left(\forall n\in\bbN,\;\mu_L\Big(B_\omega\cap\Big[T-\frac1n,T\Big]\Big)>0\right)>0,
	\end{equation}
	where
	\begin{equation}\label{eq:131121b2}
	B=\{(\omega,t)\in\Omega\times[0,T]:\delta(\omega)>0\text{ and }\rho_t(\omega)\le-\delta(\omega)\}\quad(\in\cF_T\otimes\cB([0,T])).
	\end{equation}
	Then overjumping zero or premature closure is optimal.
\end{propo}

\begin{discussion}\label{disc:131121a1}
(a) The meaning of \eqref{eq:131121b1} is that, with positive probability,
resilience $\rho$ is assumed to be negative with positive Lebesgue measure in any neighbourhood of the terminal time $T$.

\smallskip
(b) It is instructive to compare \cref{lem:negresclosetoterminaltime} with part~(iii) of \cref{lem:positiveresilience}.
The assumptions are ``almost'' complementary:
compare \eqref{eq:131121a1} with \eqref{eq:131121b1}--\eqref{eq:131121b2}.
In both cases, we step a little away from $0$
(this is the role of $\delta$ in \eqref{eq:131121a1} and \eqref{eq:131121b2})
but in a ``soft'' sense (the bound $\delta$ can depend on $\omega$).

\smallskip
(c) In view of (a) and (b) we informally summarize part~(iii) of \cref{lem:positiveresilience} and \cref{lem:negresclosetoterminaltime} as follows.
Positive resilience implies that neither overjumping zero nor premature closure is optimal;
negative resilience ``close to $T$\,'' implies optimality of overjumping zero or premature closure.
There arises the question of whether negative resilience ``far from $T$\,'' also implies overjumping zero or premature closure.
The answer is negative: see \cref{ex:betasmaller1althoughnegres} below.
\end{discussion}

\begin{proof}[Proof of \cref{lem:negresclosetoterminaltime}]
\textbf{1.}
In the first step of the proof we establish that \eqref{eq:Cgeeps}, \eqref{eq:Cbdd}, and~\eqref{eq:131121b1} imply $P\times\mu_L(C)>0$, where
$$
C=\{(\omega,t)\in\Omega\times[0,T]:\wt\beta_t(\omega)>1\}
\quad(\in\cF_T\otimes\cB([0,T])).
$$
To this end, we first recall from \cite[Lemma~8.1]{ackermann2020cadlag} that $\lim_{s\uparrow T}Y_s=Y_T$ ($=\frac12$) $P$-a.s., i.e., for the solution $(Y,M^\perp)$ of \eqref{eq:BSDEforBM}, the orthogonal to $W$ martingale $M^\perp$ does not jump at terminal time $T$. We define
\begin{align*}
M=\{(\omega,t)\in\Omega\times[0,T]:\quad&
\lim_{s\uparrow T}Y_s(\omega)=Y_T(\omega)=\frac12,\\
&Y_t(\omega)\ge0,\\[1mm]
&2\rho_t(\omega)+\mu_t(\omega)-\sigma_t^2(\omega)>0,\\[1mm]
&\max\{|\rho_t(\omega)|,|\mu_t(\omega)|\}\le\ol c\},
\end{align*}
where $\ol c$ is from \eqref{eq:Cbdd}, and notice that $M\in\cF_T\otimes\cB([0,T])$, $P\times\mu_L((\Omega\times[0,T])\setminus M)=0$.
Now we set
$$
K=B\cap M,
$$
where $B$ is from \eqref{eq:131121b2}, and observe that \eqref{eq:131121b1} holds with $B$ replaced by $K$.
As $P\times\mu_L(C)=\int_\Omega \mu_L(C_\omega)\,P(d\omega)$,
we get $P\times\mu_L(C)>0$, once we prove
\begin{equation}\label{eq:131121b3}
F:=\left\{\omega\in\Omega:\forall n\in\bbN,\;\mu_L\Big(K_\omega\cap\Big[T-\frac1n,T\Big]\Big)>0\right\}
\subseteq\{\omega\in\Omega:\mu_L(C_\omega)>0\}.
\end{equation}
To establish \eqref{eq:131121b3}, we fix an arbitrary $\omega_0\in F$ and make the following simple observation
$$
t\in K_{\omega_0}
\;\;\Longleftrightarrow\;\;
(\omega_0,t)\in M\text{ and }\rho_t(\omega_0)\le-\delta(\omega_0)<0.
$$
This yields that, for $t\in K_{\omega_0}$, it holds
$$
\mu_t(\omega_0)-\sigma_t^2(\omega_0)\le|\mu_t(\omega_0)|\le\ol c,
$$
hence
\begin{equation}\label{eq:131121b5}
0<2\rho_t(\omega_0)+\mu_t(\omega_0)-\sigma_t^2(\omega_0)
<\rho_t(\omega_0)+\mu_t(\omega_0)-\sigma_t^2(\omega_0)
\le\ol c-\delta(\omega_0).
\end{equation}
Now we compute from \eqref{eq:wtbetawoZ} that, for $t\in K_{\omega_0}$, we have the equivalence
\begin{equation}\label{eq:131121b6}
\wt\beta_t(\omega_0)>1
\;\;\Longleftrightarrow\;\;
2Y_t(\omega_0)>1+\frac{\rho_t(\omega_0)}{\rho_t(\omega_0)+\mu_t(\omega_0)-\sigma_t^2(\omega_0)}.
\end{equation}
Moreover, \eqref{eq:131121b5} and \eqref{eq:131121b6} reveal that, for $t\in K_{\omega_0}$,
\begin{equation}\label{eq:131121b7}
2Y_t(\omega_0)>1-\frac{\delta(\omega_0)}{\ol c-\delta(\omega_0)}
\;\;\Longrightarrow\;\;
\wt\beta_t(\omega_0)>1
\quad(\Longleftrightarrow t\in C_{\omega_0}).
\end{equation}
Recalling that $\omega_0\in F$, the definition of the event $F$ in \eqref{eq:131121b3}, and that $\lim_{s\uparrow T}Y_s(\omega_0)=\frac12$
(as $\omega_0\in F$ implies that there exists $t\in[0,T]$ with $(\omega_0,t)\in K\subseteq M$),
we conclude from \eqref{eq:131121b7} that there exists $n_0\in\bbN$ (which depends on $\omega_0$) such that
$$
K_{\omega_0}\cap\Big[T-\frac1{n_0},T\Big]\subseteq C_{\omega_0},
$$
hence $\mu_L(C_{\omega_0})\ge\mu_L(K_{\omega_0}\cap[T-1/{n_0},T])>0$.
We thus proved \eqref{eq:131121b3} and completed the first step of the proof.

\smallskip
\textbf{2.}
The first step together with \eqref{eq:exbeta} yields $P\times\mu_L(\beta_.>1)>0$. Define the stopping time
$\tau=T\wedge\inf\{t\in[0,T]:\beta_t>1\}$ (as usual, $\inf\emptyset:=\infty$).
As $P\times\mu_L(\beta_.>1)>0$, we get, by the Fubini argument, that $P(\tau<T)>0$.
Since $\beta_{0-}=0$ and $\beta$ is c\`adl\`ag, $P$-a.s.\ on $\{\tau<T\}$ it holds
$\beta_{\tau-}\le1$ and $\beta_\tau\ge1$, which yields the result.
\end{proof}

\section{Case studies on the effects of negative resilience}\label{sec:cs}

In this section we analyze the effects of negative resilience and discuss the results of \Cref{lem:positiveresilience} and \Cref{lem:negresclosetoterminaltime} in several subsettings of \Cref{sec:prob_form}.

\subsection{A case study with piecewise constant resilience and deterministic optimal strategies}\label{sec:examplesoverjumpingzerosigmazero}

In this subsection we assume that there are $N$ different regimes of resilience. That is to say that $\rho$ is piecewise constant. Moreover, we assume that $\rho$ is deterministic, $\mu>0$ is constant and $\sigma\equiv 0$. These assumptions lead to deterministic optimal strategies. We summarize the results in the following proposition.

\begin{propo}\label{prop:piecewise_const_res}
Assume that $\gamma_0>0$ is deterministic and that\footnote{This assumption is only for ease of exposition. All statements hold also in the case $x-\frac{d}{\gamma_0}<0$ with the suitable adjustments.} $x-\frac{d}{\gamma_0}>0$. Suppose furthermore that $\sigma\equiv 0$, that $\mu>0$ is a deterministic constant, and that $\rho \colon [0,T]\to (-\mu/2,\infty)$ is piecewise constant in the sense that there exist $N\in \N$, $\rho^{(1)},\ldots, \rho^{(N)} \in (-\mu/2,\infty)$, and $0=T_0<T_1<\ldots<T_N=T$ such that for all $t\in [0,T)$ it holds 
$
\rho_t=\sum_{i=1}^N\rho^{(i)}\mathbf{1}_{[T_{i-1},T_i)}(t).
$
Then, \eqref{eq:Cgeeps} and \eqref{eq:Cbdd} are satisfied. The unique solution of \eqref{eq:BSDEforBM} is given by
\begin{equation}\label{eq:251121a1}
Y_t  = e^{(T-t)\mu} 
		\Bigg( 
		\sum_{i=n(t)+1}^N \frac{(\rho^{(i)}+\mu)^2e^{T\mu}}{\mu(\rho^{(i)} + \frac12\mu)}
		(e^{-(t\vee T_{i-1})\mu}-e^{-T_i\mu}) + 2
		\Bigg)^{-1}, \quad M^\perp_t=0, \quad  t \in [0,T],
\end{equation}
where $n(t)=\max\{i\in \{0,\ldots, N\}\colon T_i\le t\}$. Moreover, \eqref{eq:exbeta} is satisfied with 
$
	\beta_t = \wt\beta_t = \frac{\rho_t + \mu}{\rho_t + \frac12\mu} Y_t$, $t \in [0,T].
$
The optimal strategy $X^*$ and the associated deviation $D^*$ are deterministic, for every $i\in \{1,\ldots, N\}$ they are continuous on $(T_{i-1},T_i)$, and for every $i\in \{1,\ldots, N-1\}$ they have a jump at $T_i$ if and only if $\rho$ has a jump at $T_i$.
Furthermore, for every $i\in \{1,\ldots, N\}$ the deviation $D^*$ is constant on $(T_{i-1},T_i)$ and takes negative values, and the optimal strategy $X^*$ is monotone on $(T_{i-1},T_i)$: more precisely, if $\rho^{(i)}>0$ (resp., $\rho^{(i)}<0$; resp., $\rho^{(i)}=0$), then $X^*$ is strictly decreasing (resp., strictly increasing; resp., constant) on $(T_{i-1},T_i)$.
\end{propo}

\begin{proof}
Clearly, \eqref{eq:Cgeeps} and \eqref{eq:Cbdd} are satisfied. Next note that
$Y$ from \eqref{eq:251121a1} satisfies for all $t\in [0,T]$ that
\begin{equation*}
	\begin{split}
		Y_t & = e^{(T-t)\mu} 
		\Bigg( 
		\int_t^T \frac{(\rho_s+\mu)^2}{\rho_s + \frac12\mu} e^{(T-s)\mu} ds + 2
		\Bigg)^{-1}.
	\end{split}
\end{equation*}
From this it follows that $Y$ is continuous and satisfies the Bernoulli ODE
\begin{equation*}
	d Y_t = \frac{(\rho_t + \mu)^2}{\rho_t + \frac12 \mu } Y_t^2 dt - \mu Y_t dt, \quad t \in [0,T], \quad Y_T=\frac12.
\end{equation*}
Consequently, $(Y,0)$ is the unique solution of \eqref{eq:BSDEforBM}. Moreover, $\wt \beta$ defined by \eqref{eq:wtbetawoZ} is c\`adl\`ag and of finite variation and thus we have \eqref{eq:exbeta} with $\beta=\wt\beta$. In particular, $\beta$ is deterministic, and since $\sigma \equiv 0$ and $\rho,\mu$ are deterministic, we have that the optimal strategy $X^*$ and its deviation $D^*$ are deterministic as well. 

For every $i\in \{1,\ldots, N-1\}$ observe also that $\beta$ has a jump at $T_i$ if and only if $\rho$ has a jump at $T_i$. This directly translates into jumps of the optimal strategy $X^*$ and jumps of the associated deviation $D^*$ via \eqref{eq:optstrat} and \eqref{eq:optdev}.
To show that the deviation $D^*$ is constant on each $(T_{i-1},T_i)$, $i \in \{1,\ldots, N\}$, observe that for all $i \in \{1,\ldots,N\}$ and $t \in (T_{i-1},T_i)$ it holds that 
\begin{equation}\label{eq:dbeta}
	\begin{split}
		d\beta_t & = \frac{\rho^{(i)} + \mu}{\rho^{(i)} + \frac12 \mu} \left( \frac{(\rho^{(i)} + \mu)^2}{\rho^{(i)} + \frac12 \mu} Y_t^2 - \mu Y_t \right) dt 
		= \beta_t^2 (\rho^{(i)} + \mu) dt - \mu \beta_t dt
	\end{split}
\end{equation}
and hence
\begin{equation*}
	\begin{split}
		d\left( \gamma_t \beta_t \cE(Q)_t \right) 
		& = \beta_t d(\gamma_t \cE(Q)_t) 
		+ \gamma_t \cE(Q)_t d\beta_t \\
		& = \beta_t \gamma_t \cE(Q)_t \left( \mu - \beta_t (\mu + \rho^{(i)}) \right) dt  
		+ \gamma_t \cE(Q)_t \left( \beta_t^2 (\rho^{(i)} + \mu) - \mu \beta_t \right) dt 
		\\
		& = 0. 
	\end{split}
\end{equation*}
It thus follows from \eqref{eq:optdev} that $D^*$ is constant on $(T_{i-1},T_i)$ for $i \in \{1,\ldots,N\}$.  
Moreover, since $\rho>-\frac12\mu$, $\mu>0$, and $Y>0$, it holds that $\beta> 0$, and therefore $D^* < 0$ (recall that we assume $x-\frac{d}{\gamma_0}>0$). 
Next note that
we have for all $i \in \{1,\ldots, N\}$ and $t \in (T_{i-1},T_i)$ that, using \eqref{eq:dbeta},  
\begin{equation*}
	\begin{split}
		d((1-\beta_t)\cE(Q)_t) & = (1-\beta_t) d\cE(Q)_t - \cE(Q)_t d\beta_t \\
		& = - \cE(Q)_t (1-\beta_t) \beta_t \left( \mu + \rho^{(i)} \right) dt 
		- \cE(Q)_t \left( \beta_t^2 (\rho^{(i)} + \mu) - \mu \beta_t \right) dt \\
		& = -\cE(Q)_t \beta_t \rho^{(i)} dt .
	\end{split}
\end{equation*}
Since $\beta > 0$ and $x-\frac{d}{\gamma_0}>0$, we conclude that if $\rho^{(i)}$ is positive, then $X^*$ in \eqref{eq:optstrat} is decreasing on $(T_{i-1},T_i)$, and if $\rho^{(i)}$ is negative, then $X^*$ is increasing on $(T_{i-1},T_i)$, $i \in\{1,\ldots, N\}$.
\end{proof}

In \Cref{ex:betasmaller1althoughnegres}, \Cref{ex:prematureclosureonlyattimepoint}, and \Cref{ex:rholess0equalsbetagreater1} below we consider the setting of \Cref{prop:piecewise_const_res} with $N=3$ different regimes of resilience. More precisely, we assume in the sequel of this subsection the setting of \Cref{prop:piecewise_const_res} with $N=3$, $x=1$, $d=0$, $\gamma_0=1$, $\mu=0.5$, and $T_i=i$ for $i\in\{1,2,3\}$.

We already know from \Cref{lem:negresclosetoterminaltime} that overjumping zero or premature closure is optimal if we have negative resilience in the last regime (i.e., $\rho^{(3)}<0$). In the three examples below we want to analyze under which conditions these effects occur in the case where the resilience is positive in the last (and also the first) regime. 
We choose $\rho^{(1)}=0.1$ and $\rho^{(3)}=1$. \Cref{lem:positiveresilience} entails that we necessarily need $\rho^{(2)}<0$ to see these effects. Therefore we choose a different negative value for $\rho^{(2)}$ in each example.

For these choices of $\rho^{(i)}$, $i\in\{1,2,3\}$, \Cref{prop:piecewise_const_res} shows that it is optimal to first sell during $(0,1)$, change this to a buy program on $(1,2)$ to profit from the negative resilience during that time interval, and then sell again during $(2,3)$. 
Moreover, since $\rho^{(1)}$ and $\rho^{(3)}$ are positive, we can already derive (e.g., by \Cref{lem:positiveresilience}) that $\beta<1$ on $[0,1)$ and on $[2,3)$, and hence that $X^*$ is strictly positive on $[0,1)$ and on $[2,3)$ due to $x-\frac{d}{\gamma_0}=1$. 
Between
\Cref{ex:betasmaller1althoughnegres},  \Cref{ex:prematureclosureonlyattimepoint},
and 
\Cref{ex:rholess0equalsbetagreater1}
we vary the size of $\rho^{(2)}<0$. 
This then determines if we get overjumping zero or premature closure for the optimal strategy. 
Recall that $\beta$ in all examples has jumps at $t=1$ and $t=2$ and is continuous on $(0,1)$, $(1,2)$, and $(2,3)$, with values strictly smaller than $1$ on $[0,1)$ and $[2,3)$. 
The facts that $\rho^{(1)}=0.1$, $\mu=0.5$, and $Y_1\in (0,1/2]$ yield that also $\beta_{1-}<1$.
We moreover have that $(1-\beta_{t-})(1-\beta_t)>0$ for all $t \in [0,1)\cup (2,3)$.
This, continuity of $\beta$ on $(1,2)$, $\beta_{1-}<1$, and $\beta_2<1$ imply that overjumping zero is optimal if and only if at least one of 
\begin{equation}\label{eq:condoverjump1}
	\beta_1=\frac{\rho^{(2)}+\frac12}{\rho^{(2)}+\frac14}Y_1>1
\end{equation}
and 
\begin{equation}\label{eq:condoverjump2}
	\beta_{2-}=\frac{\rho^{(2)}+\frac12}{\rho^{(2)}+\frac14}Y_2>1
\end{equation}
is satisfied.
Premature closure is optimal if and only if 
\begin{equation}\label{eq:condpremclosure}
	\text{there exists } t \in [1,2] \text{ such that } (1-\beta_{t-})(1-\beta_t)=0.
\end{equation}
The resilience $\rho$, the function $\beta$, and the optimal strategy $X^*$ for each of the examples below are shown in \Cref{fig:exoverjumpingzerosigmazero}.

\begin{figure}[!htb]
	\centering
	\setkeys{Gin}{width=0.346\linewidth}
	\includegraphics{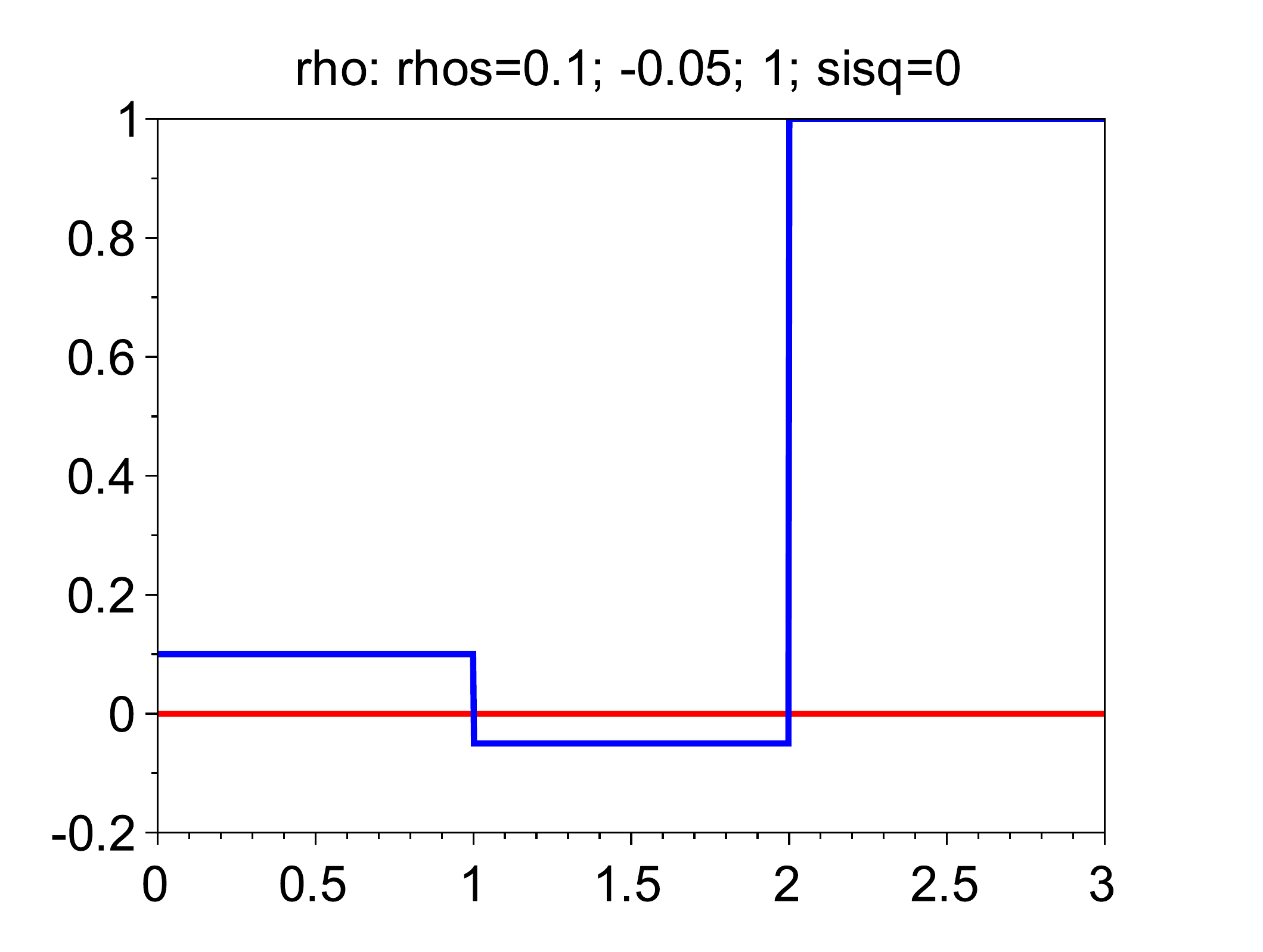}\,	\hspace{-0.5cm}\includegraphics{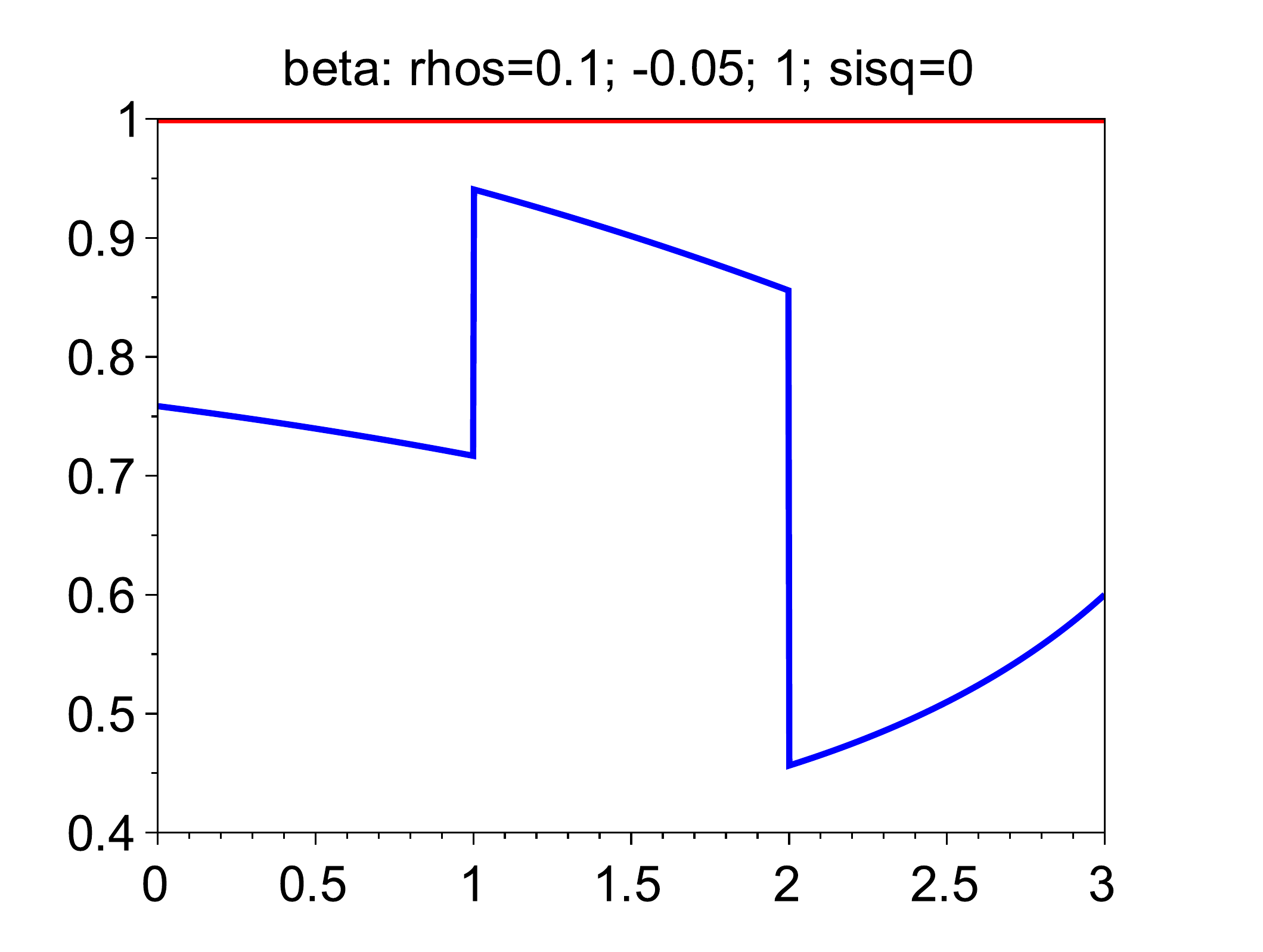}\,	\hspace{-0.5cm}\includegraphics{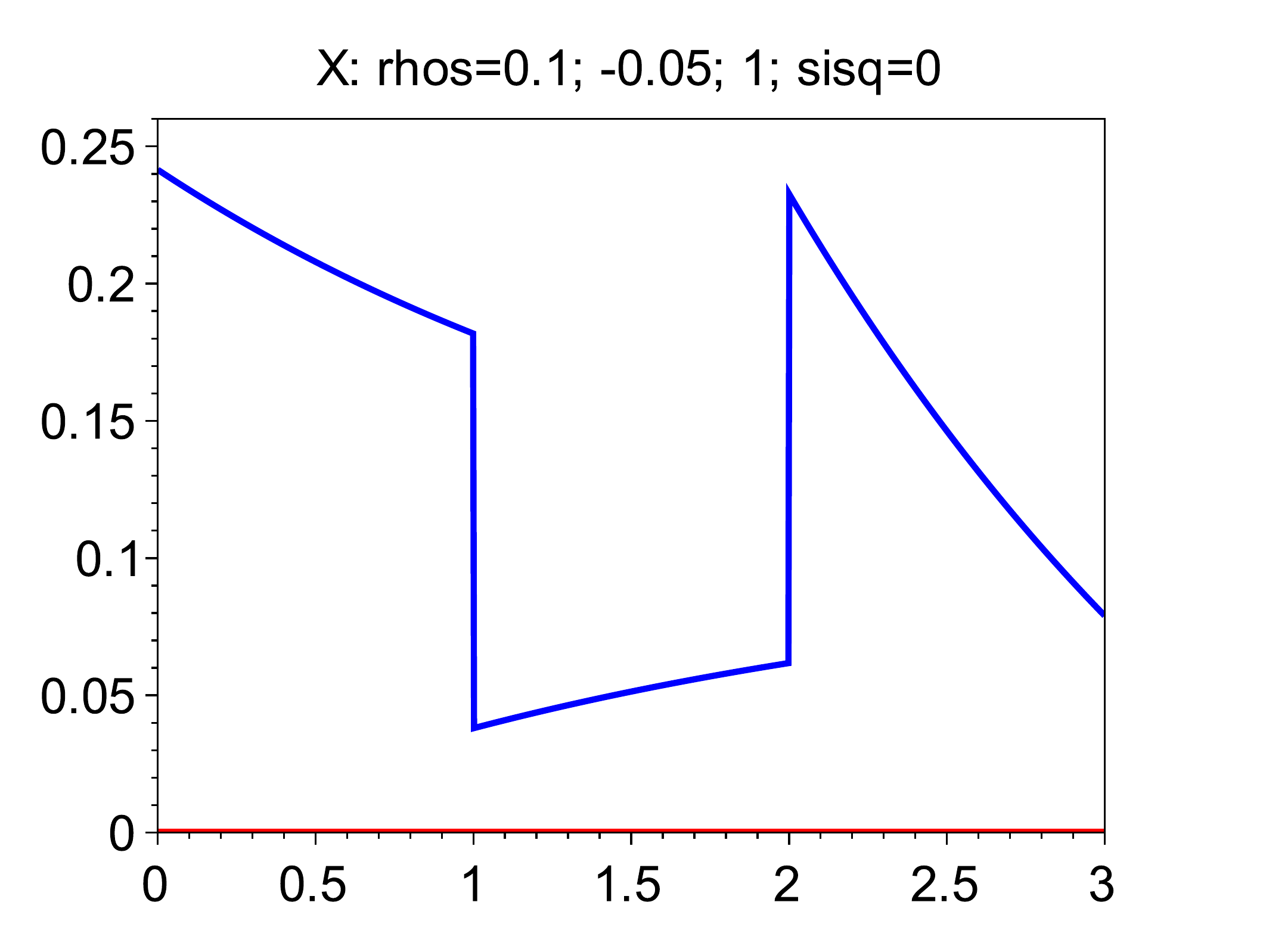} 
	
\includegraphics{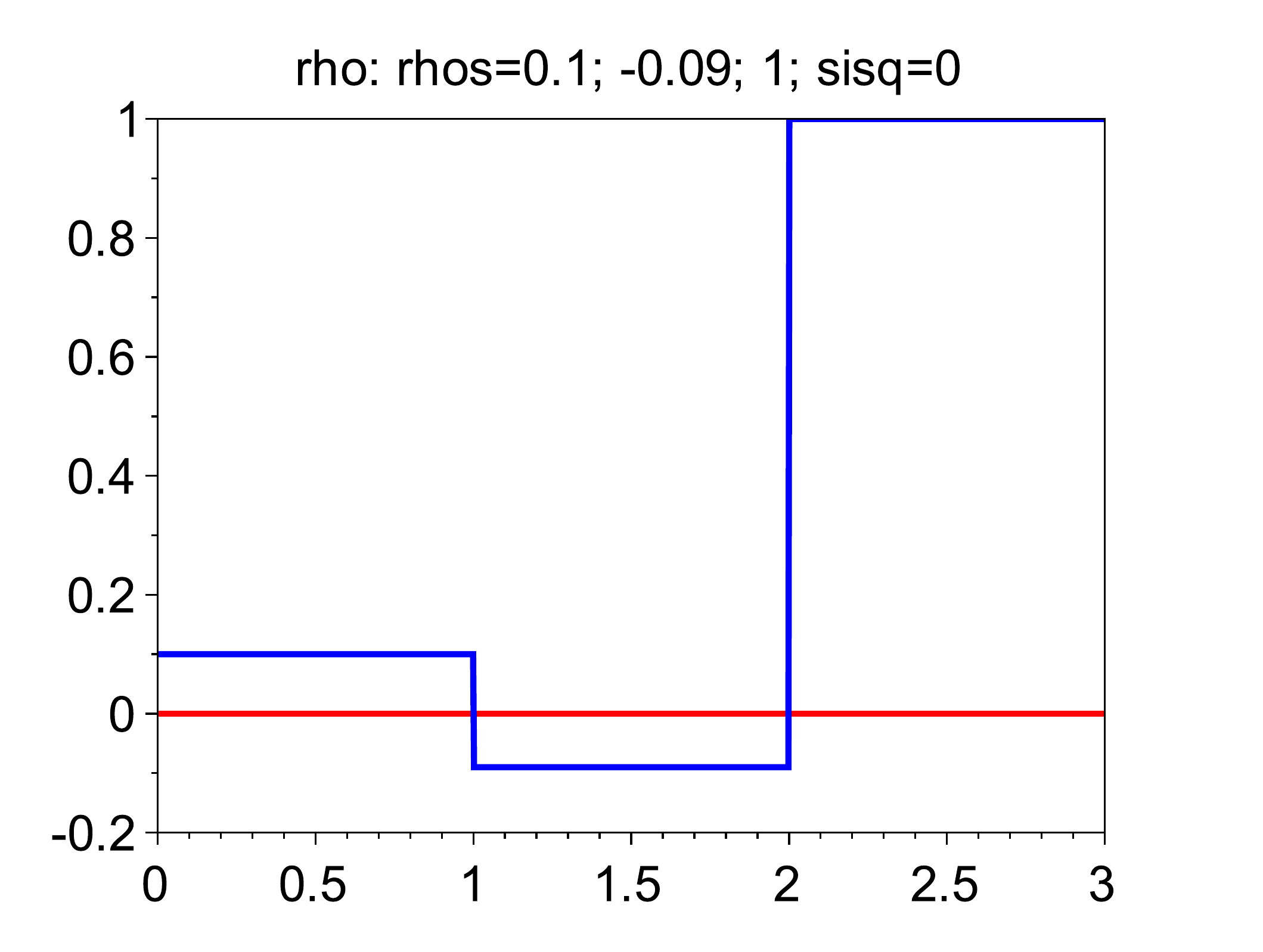}\, 	\hspace{-0.5cm}\includegraphics{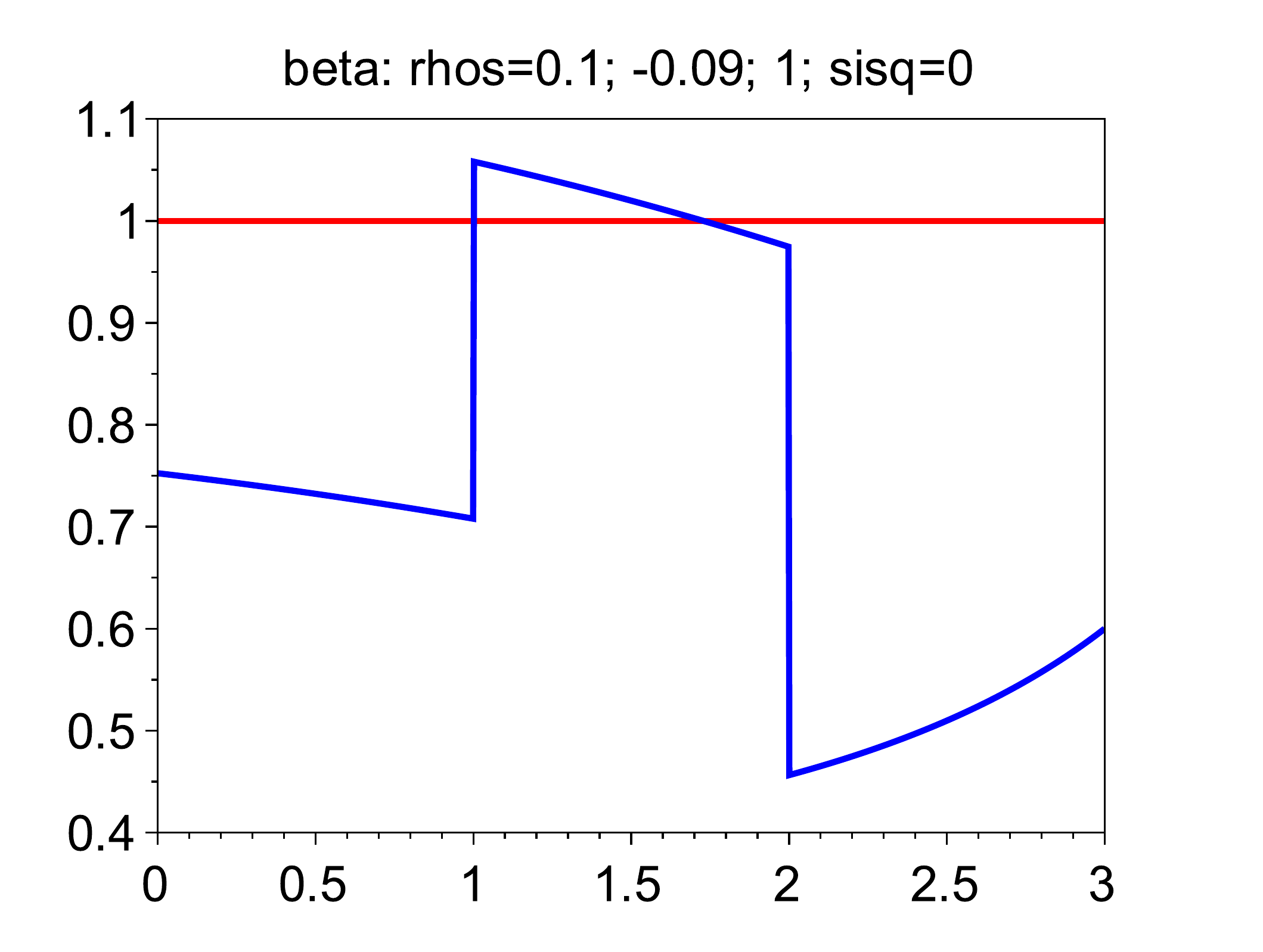}\, 	\hspace{-0.5cm}\includegraphics{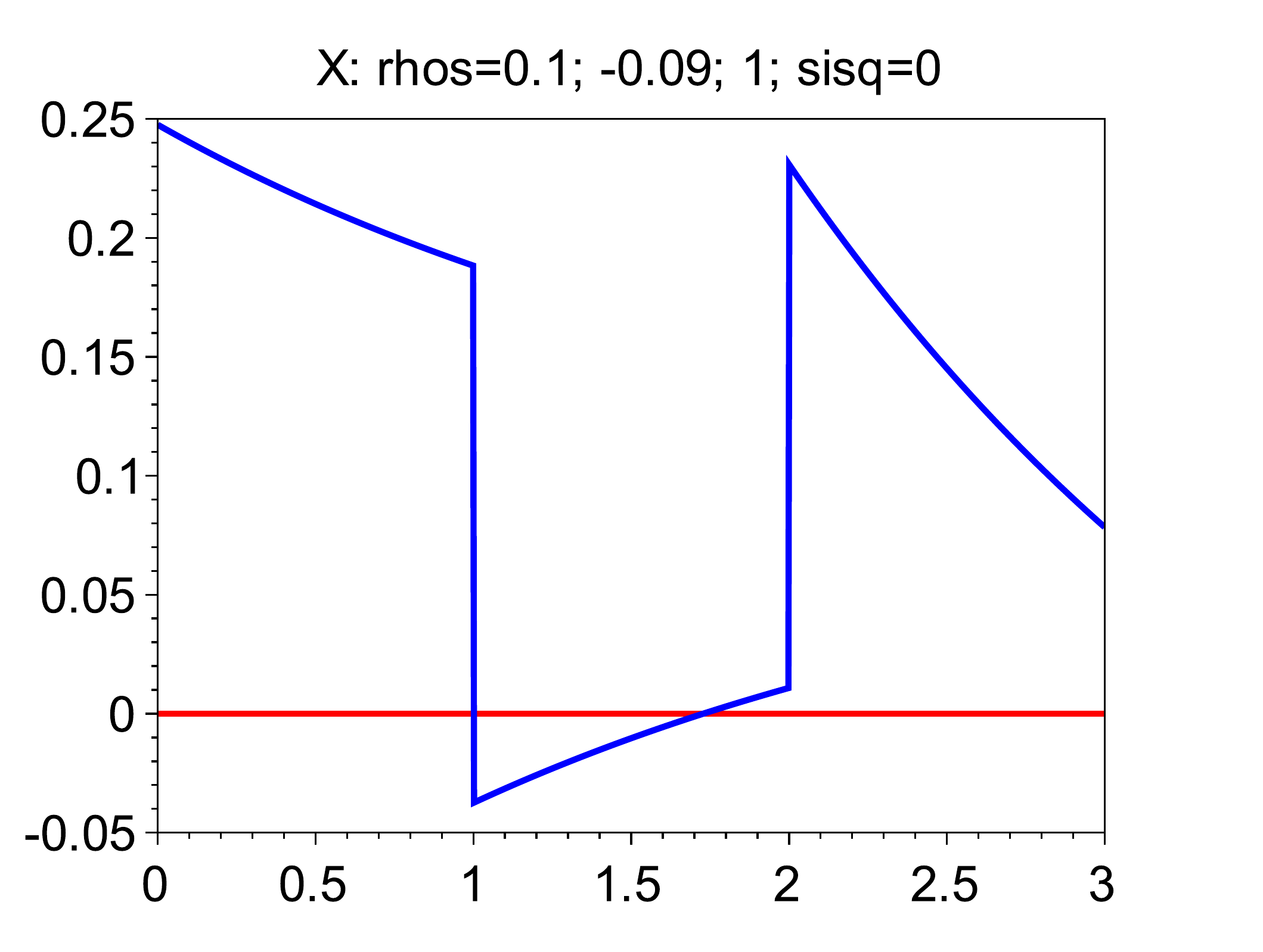}	
	
	\includegraphics{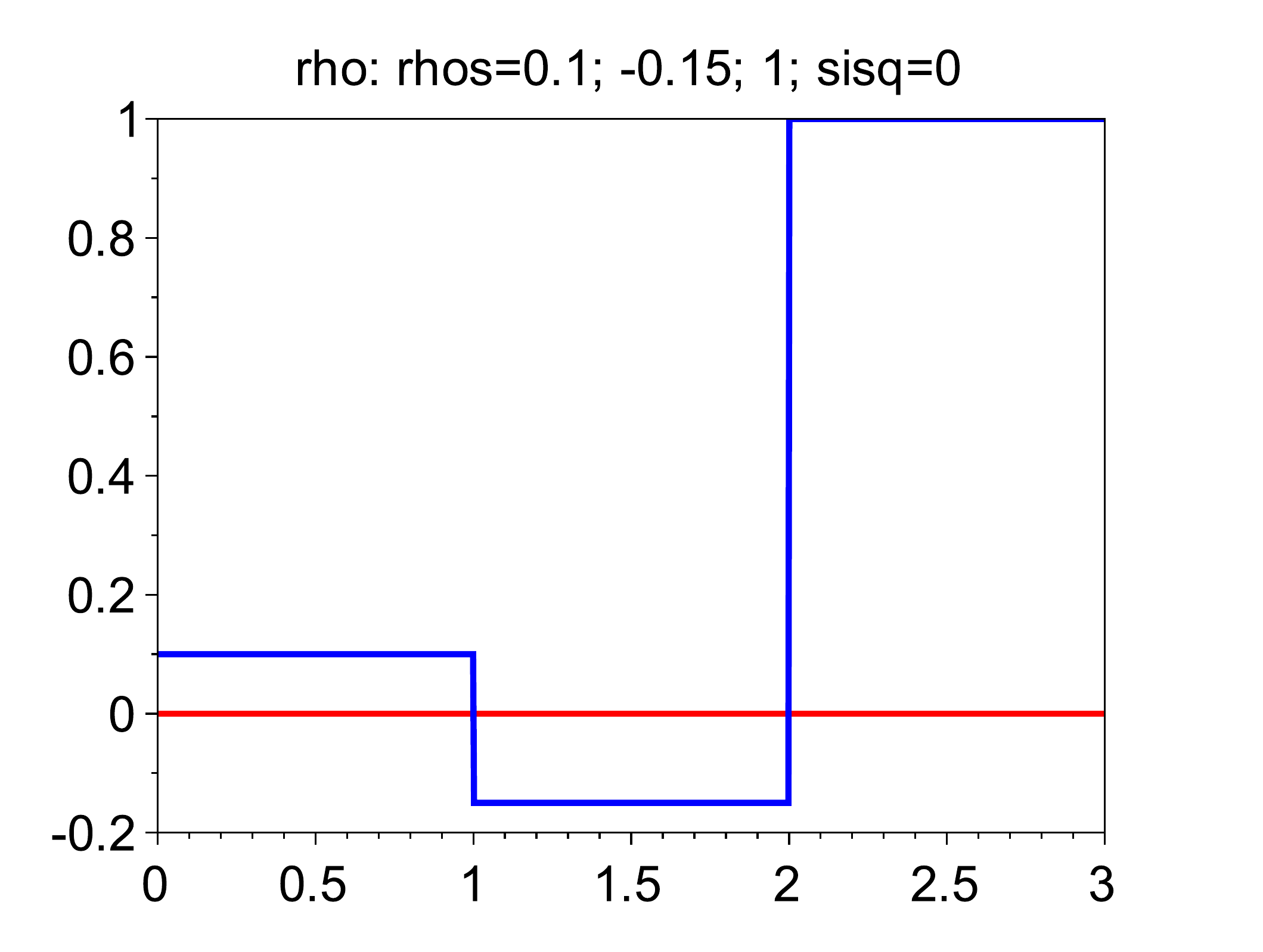}\,	\hspace{-0.5cm}\includegraphics{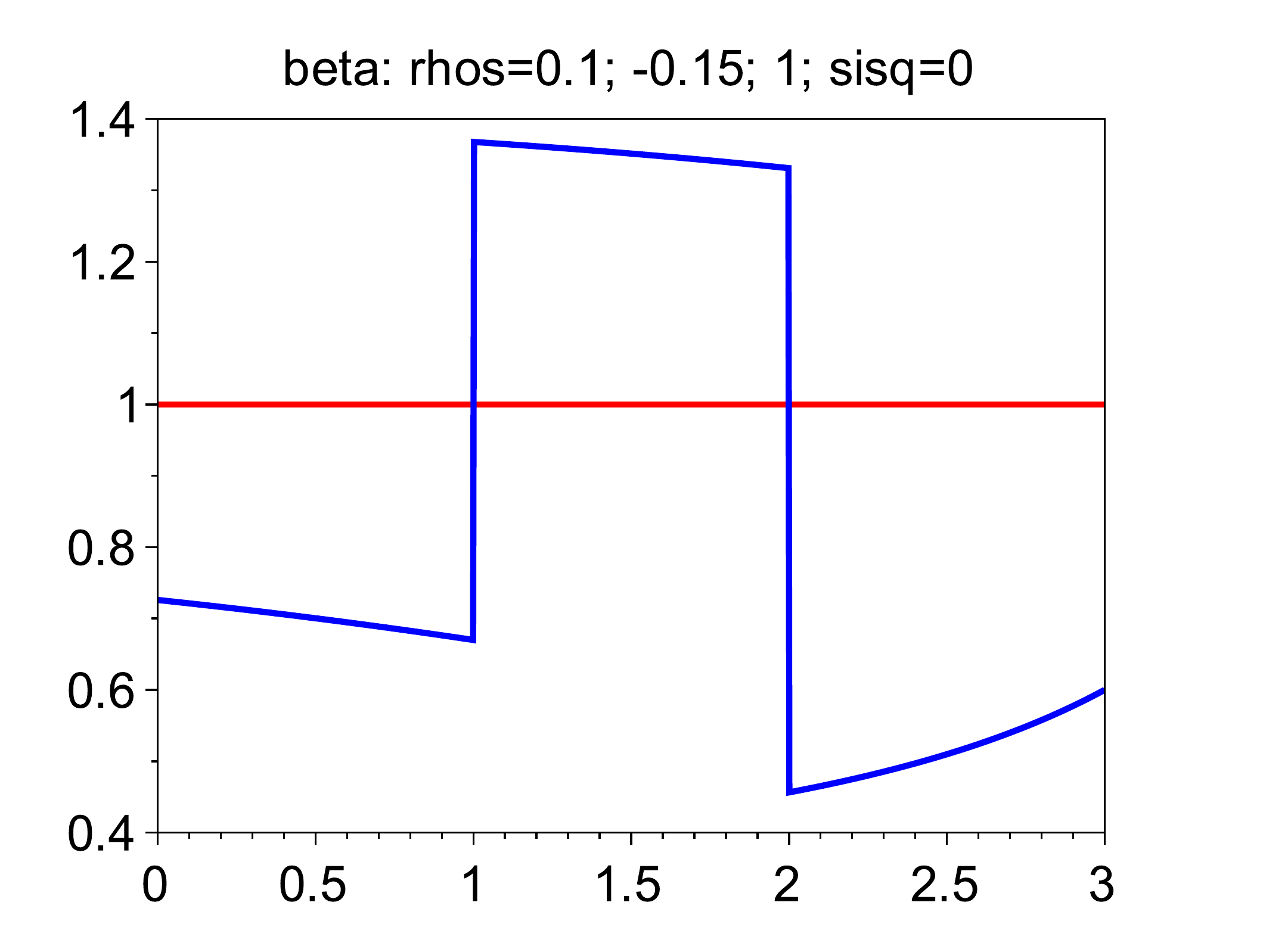}\, 	\hspace{-0.5cm}\includegraphics{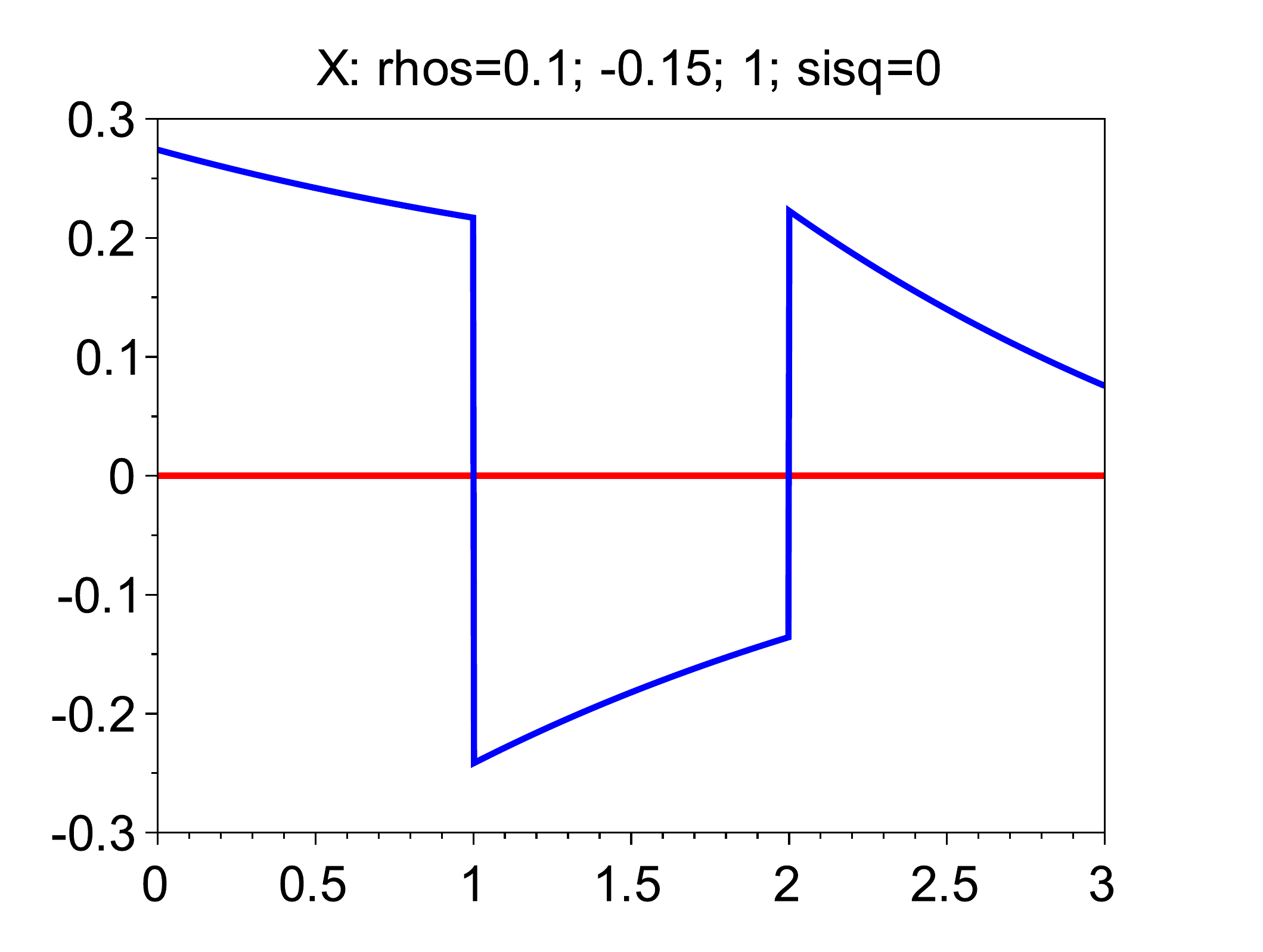} 
	 		
	\caption{Top row: $\rho$, $\beta$, and $X^*$ in \Cref{ex:betasmaller1althoughnegres}. Middle row: $\rho$, $\beta$, and $X^*$ in \Cref{ex:prematureclosureonlyattimepoint}. Bottom row: $\rho$, $\beta$, and $X^*$ in \Cref{ex:rholess0equalsbetagreater1}. }
	\label{fig:exoverjumpingzerosigmazero}
\end{figure}

\begin{ex}\label{ex:betasmaller1althoughnegres}
We choose $\rho^{(2)}=-0.05$. The first row in \Cref{fig:exoverjumpingzerosigmazero} shows that $\beta$ stays strictly smaller than one also on $[1,2)$, and hence the optimal strategy $X^*$ is strictly positive on the time interval $[0,3)$. 
 We conclude that, in general, a period of negative resilience does not necessarily lead to overjumping zero or premature closure. 
\end{ex}

\begin{ex}\label{ex:prematureclosureonlyattimepoint}
We next provide an example where negative resilience indeed leads to overjumping zero and premature closure. To this end we choose  $\rho^{(2)}=-0.09$ in the above set-up. From the second row of \Cref{fig:exoverjumpingzerosigmazero} we observe that $\beta$ jumps above $1$ at time $t=1$, but then decays continuously below $1$ already before its next jump at $t=2$. 
	It therefore holds that \eqref{eq:condoverjump1} and \eqref{eq:condpremclosure} are satisfied. 
	We thus have overjumping zero as well as premature closure for the optimal strategy. 
	This implies (recall $x-\frac{d}{\gamma_0}=1$) that the optimal strategy jumps to a negative value at time $t=1$ and crosses $0$ within the time interval $(1,2)$ to become positive again. Note that the set of points in time $t \in [0,T)$ for which we have $\beta_t>1$ is strictly included in the set where $\rho_t<0$ (which is $[1,2)$).
\end{ex}

\begin{ex}\label{ex:rholess0equalsbetagreater1}
We finally provide an example where the set of points in time $t \in [0,T)$ for which we have $\beta_t>1$ is equal to the set where $\rho_t<0$. 
	This means that the time periods with negative resilience exactly coincide with the time periods where the optimal strategy is negative.
	We achieve this for example for $\rho^{(2)}=-0.15$ in the above set-up (see the third row of \Cref{fig:exoverjumpingzerosigmazero}).
	In particular, \eqref{eq:condoverjump1} is satisfied, i.e., overjumping zero is optimal. 
	Furthermore, one can compute that \eqref{eq:condoverjump2} holds true as well. 
	It follows that condition \eqref{eq:condpremclosure} is not met, and therefore, premature closure is not optimal. 
	Note that the optimal strategy changes its sign twice, but does not continuously cross $0$.
\end{ex}

\subsection{A case study with piecewise constant resilience and stochastic optimal strategies}\label{sec:examplesoverjumpingzerowithsigma}

We here consider a similar setting as in \Cref{sec:examplesoverjumpingzerosigmazero}, but now $\sigma$ can be a deterministic constant different from $0$. 
Although the solution of BSDE~\eqref{eq:BSDEforBM} and the process $\beta=\wt\beta$ are still deterministic, the optimal strategy $X^*$ and its associated deviation $D^*$ in general become stochastic. 
The properties derived in \Cref{prop:piecewise_const_res} that $D^*$ is constant between jumps and that $X^*$ is monotone between jumps then no longer hold. 
However, we can produce the main effects discussed in \Cref{ex:betasmaller1althoughnegres}, \Cref{ex:prematureclosureonlyattimepoint},
and \Cref{ex:rholess0equalsbetagreater1} also in the case with nonzero $\sigma$. 
Let $\sigma= \sqrt{0.1}$, $\mu=0.5$, 
$x=1$, $d=0$, $\gamma_0=1$, $T=3$. Assume
$\rho$ as in \Cref{prop:piecewise_const_res}  with $N=3$, $T_0=0$, $T_1=1$, $T_2=2$, $T_3=T$, $\rho^{(1)}=0.1$, $\rho^{(3)}=1$, and a $\rho^{(2)}<0$ chosen appropriately for each example.
Then, for $\rho^{(2)}=-0.05$, we see that $\beta<1$ everywhere, which implies that neither overjumping zero nor premature closure is optimal (cf.\ the first row of \Cref{fig:exoverjumpingzerowithsigma}). This is just as in \Cref{ex:betasmaller1althoughnegres}. 
In order to obtain the same effect as in \Cref{ex:prematureclosureonlyattimepoint}, we consider $\rho^{(2)}=-0.07$. Then, $\{t\in[0,T)\colon \beta_t>1\}\subsetneq \{t\in[0,T)\colon\rho_t<0\}$, and $\beta$ jumps above $1$ in $t=1$ and goes through $1$ on $(1,2)$ (cf.\ the second row of \Cref{fig:exoverjumpingzerowithsigma}).
Consequently, both overjumping zero and premature closure are optimal in this case.
If we set $\rho^{(2)}=-0.15$, we observe that $\{t\in[0,T)\colon\rho_t<0\}=[1,2)=\{t\in[0,T)\colon\beta_t>1\}$ (cf.\ the third row of \Cref{fig:exoverjumpingzerowithsigma}), and that overjumping zero is optimal, but premature closure is not. This is the analogon of \Cref{ex:rholess0equalsbetagreater1}.

\begin{figure}[!htb]
	\centering
	\setkeys{Gin}{width=0.346\linewidth}
		\includegraphics{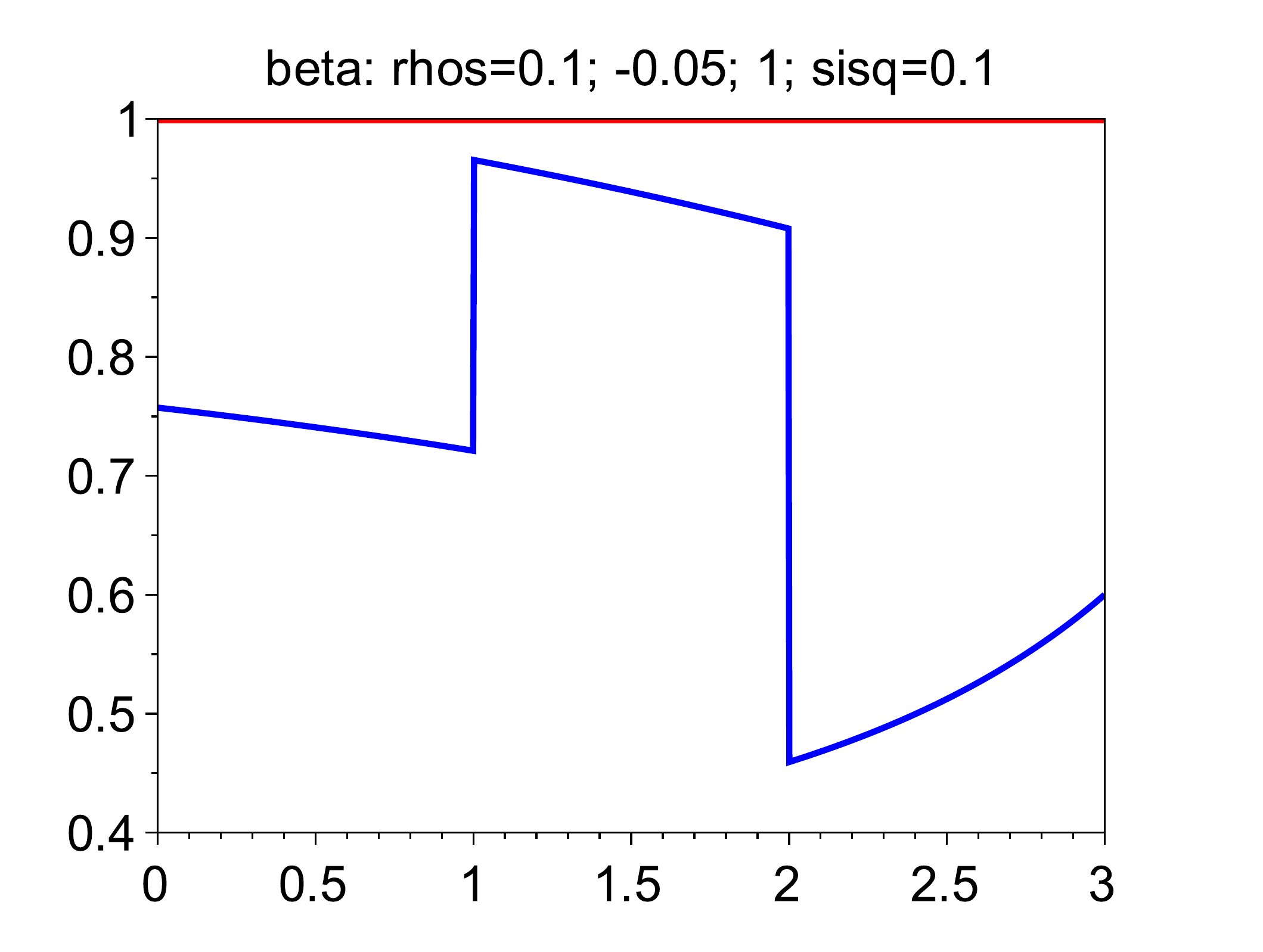}\,	\hspace{-0.5cm}\includegraphics{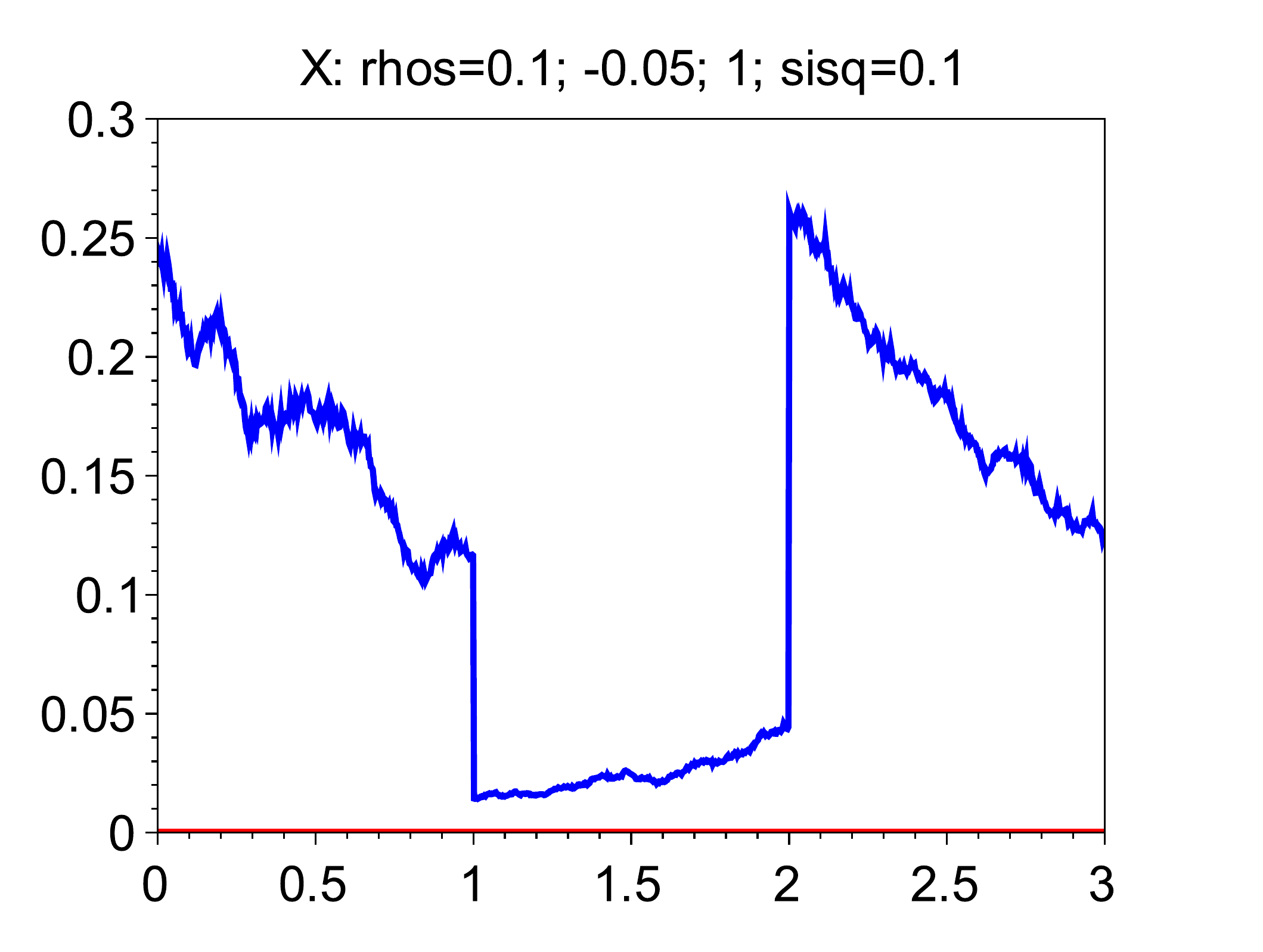}\,	\hspace{-0.5cm}\includegraphics{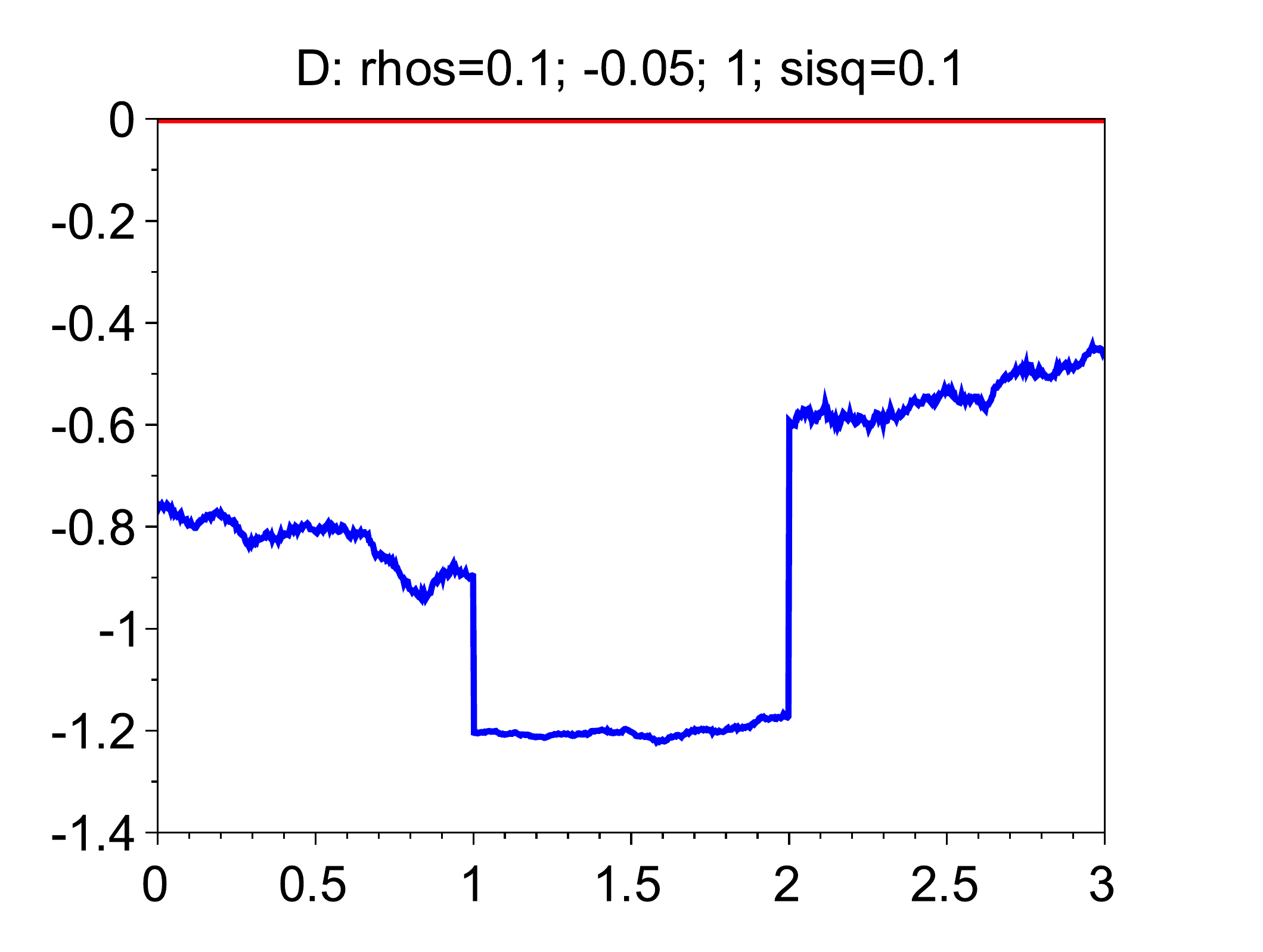} 
		
		\includegraphics{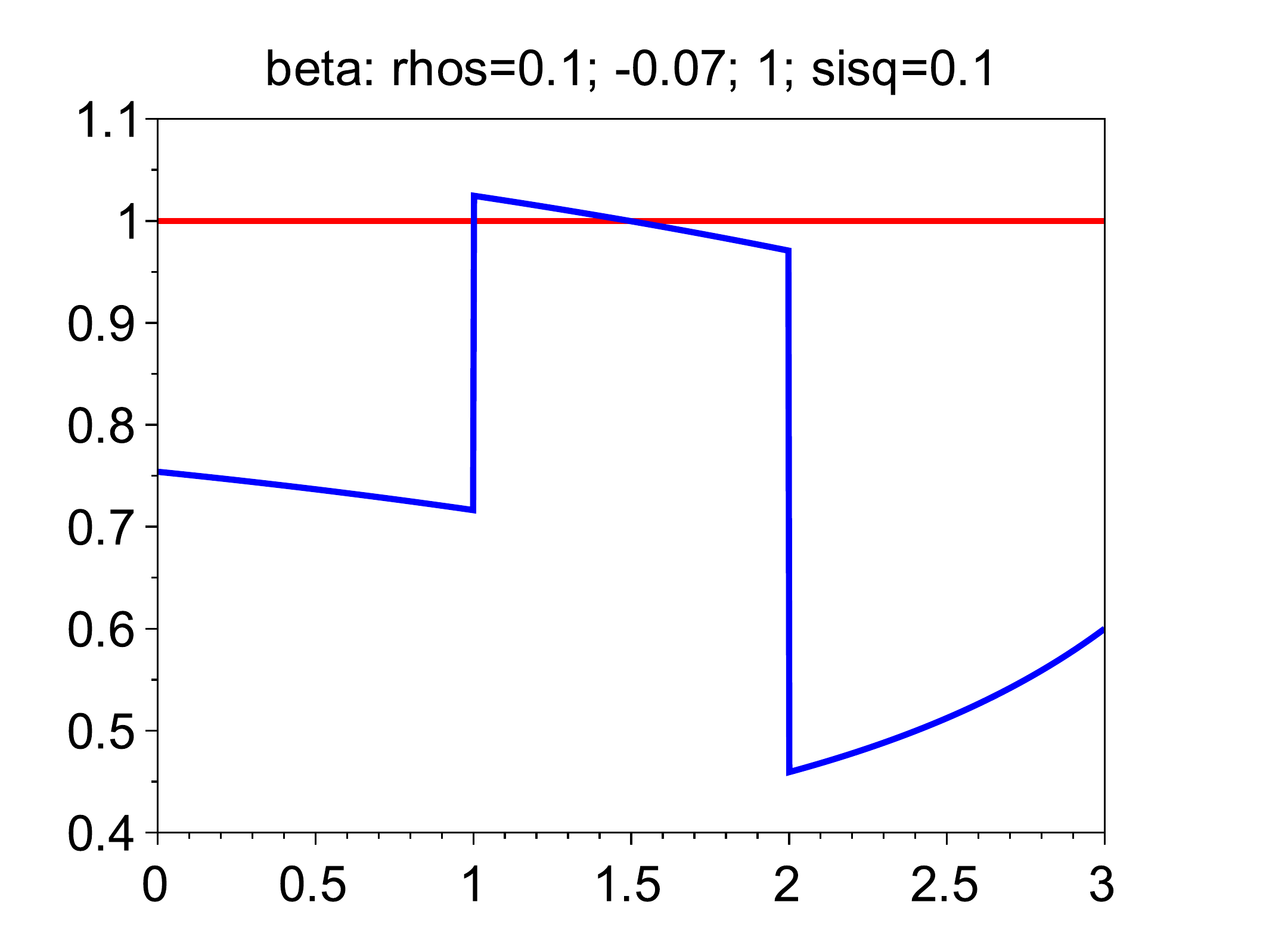}\, 	\hspace{-0.5cm}\includegraphics{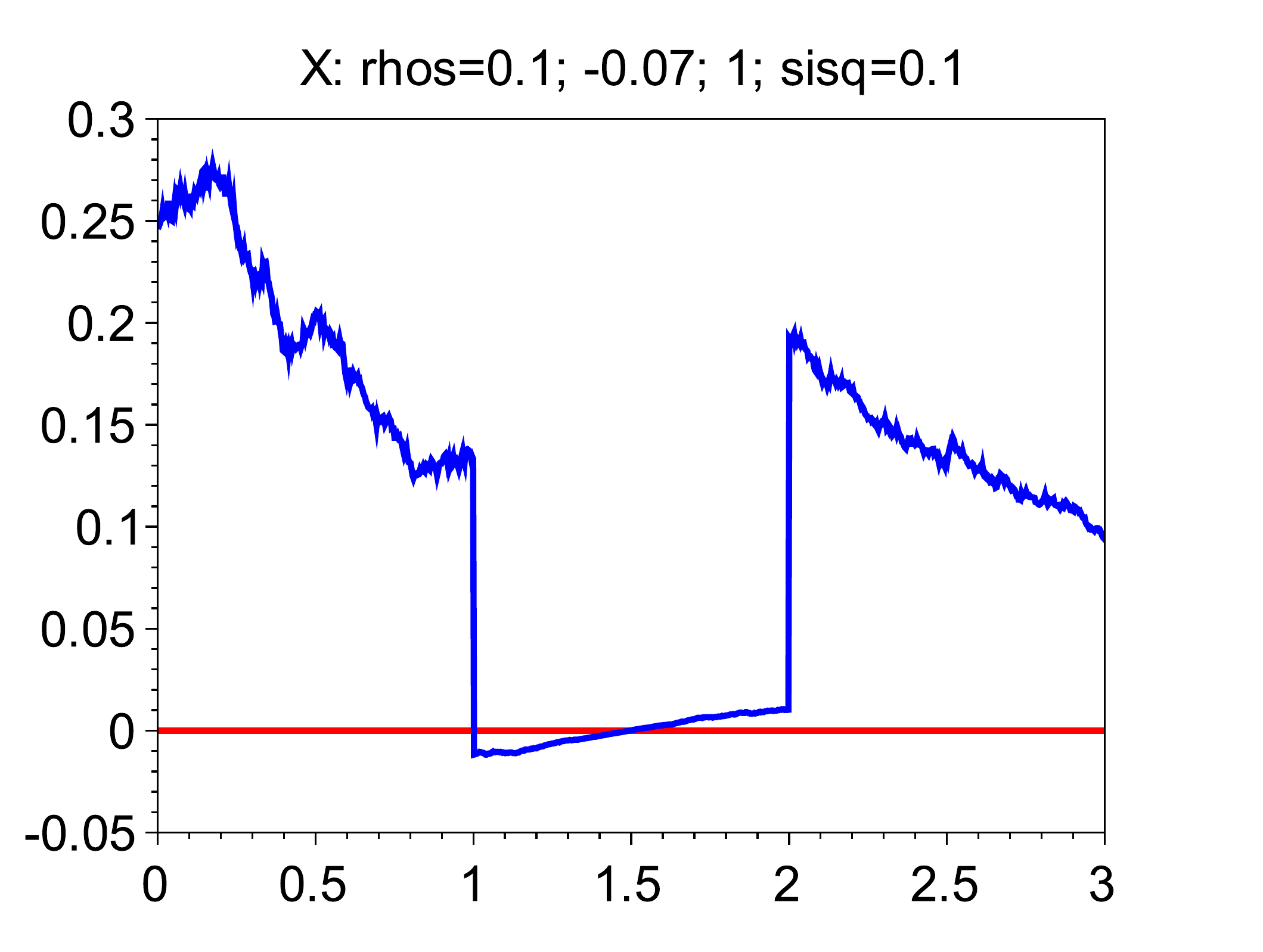}\, 	\hspace{-0.5cm}\includegraphics{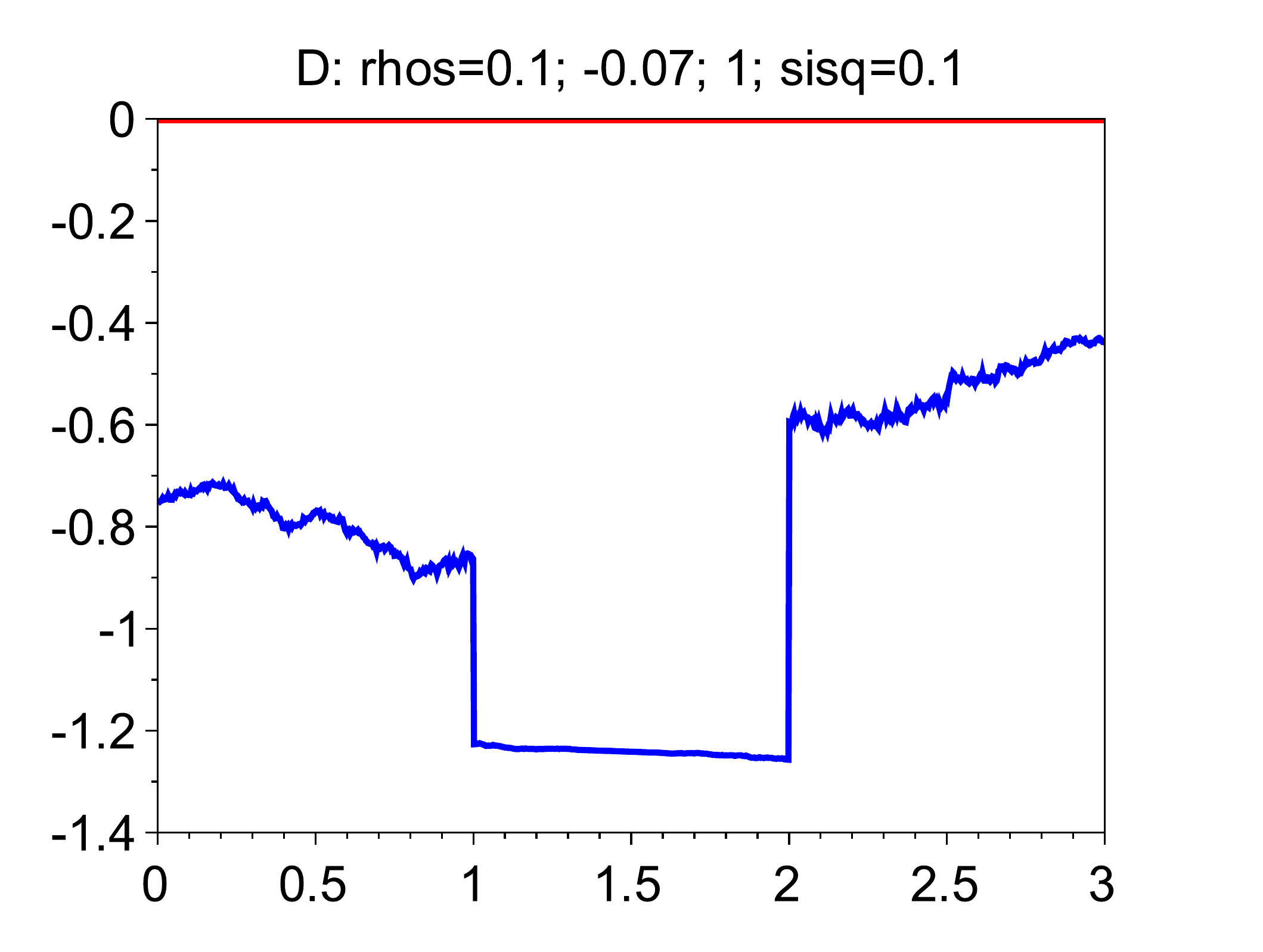} 

\includegraphics{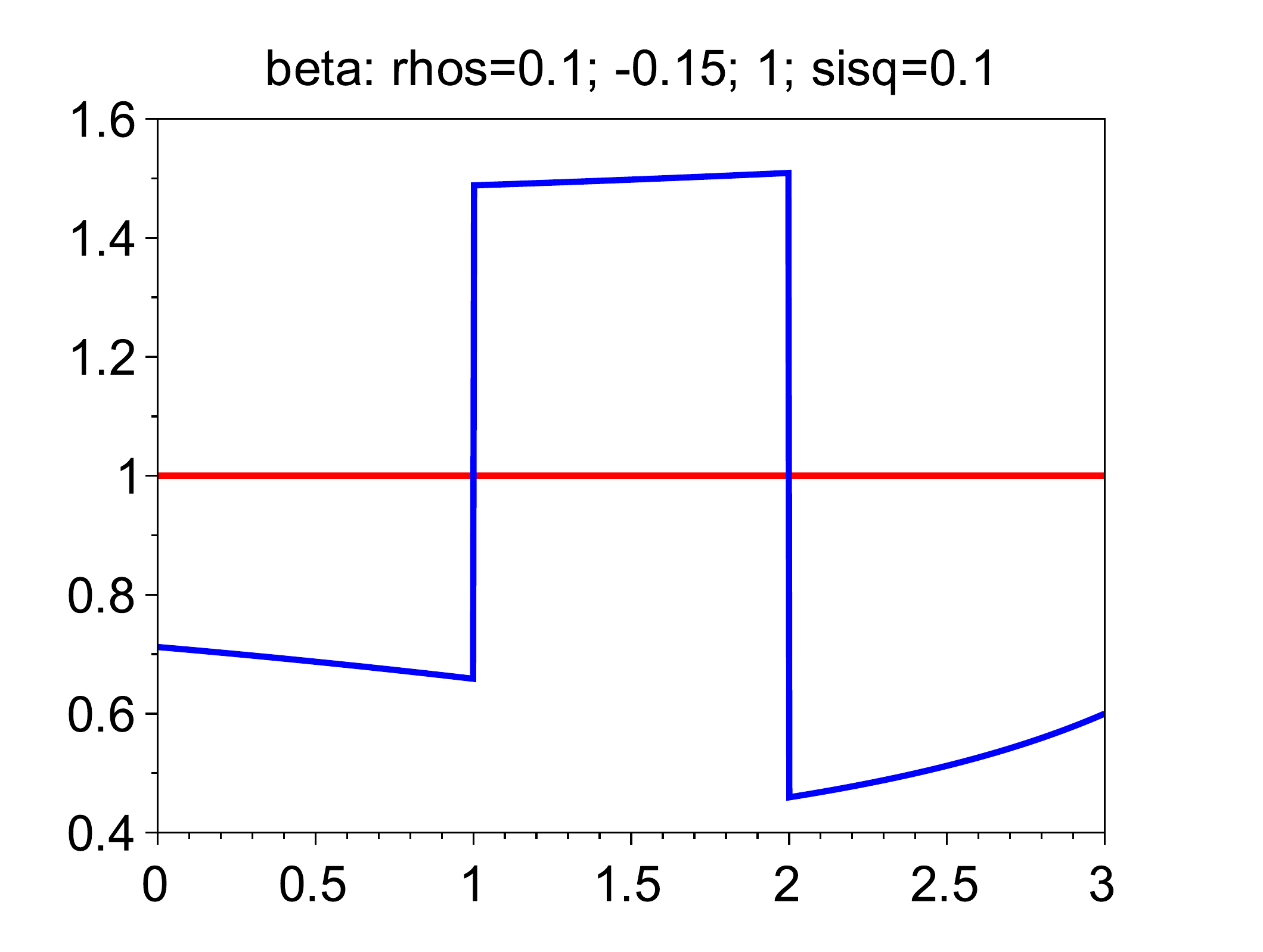}\,	\hspace{-0.5cm}\includegraphics{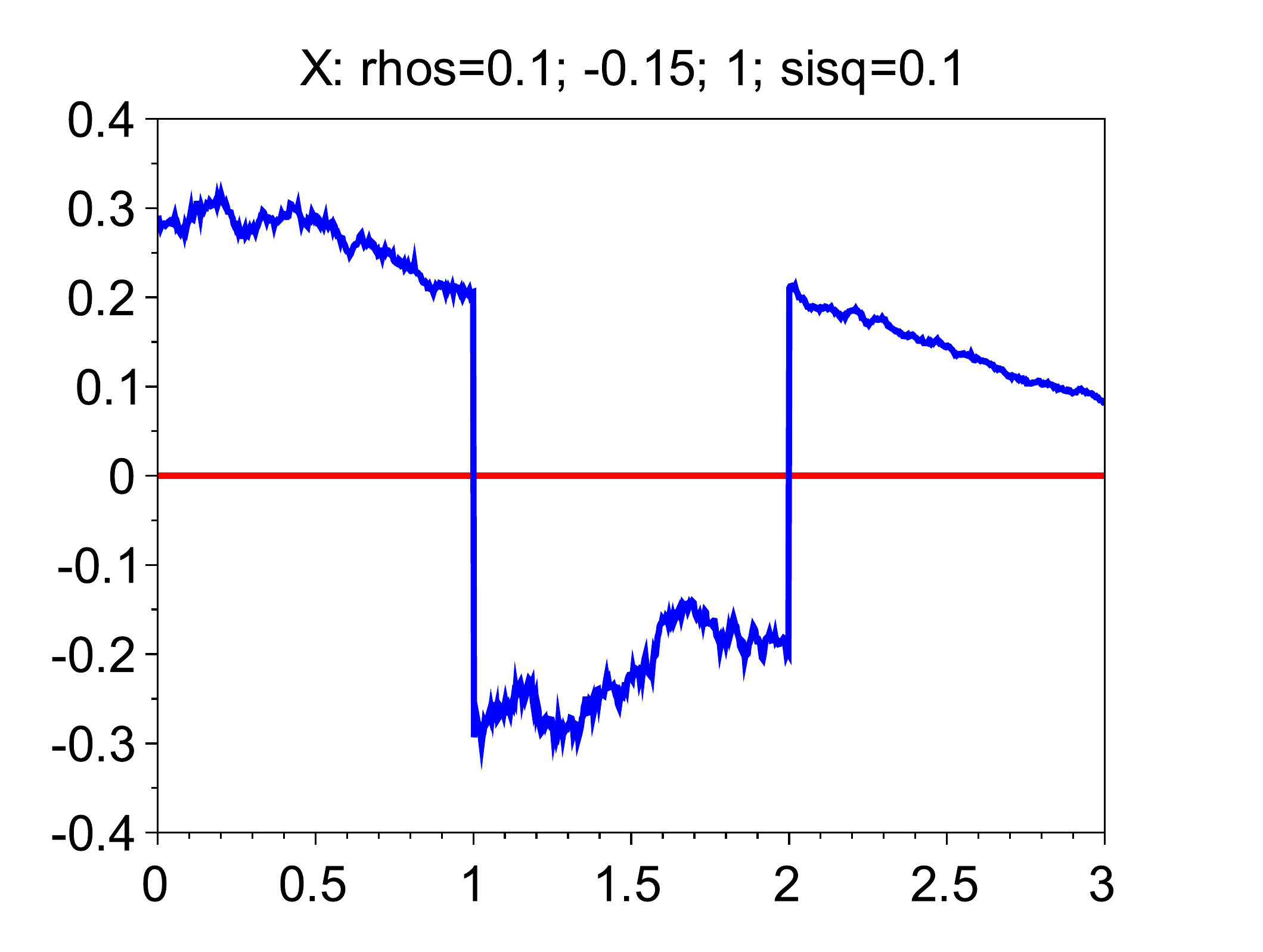}\, 	\hspace{-0.5cm}\includegraphics{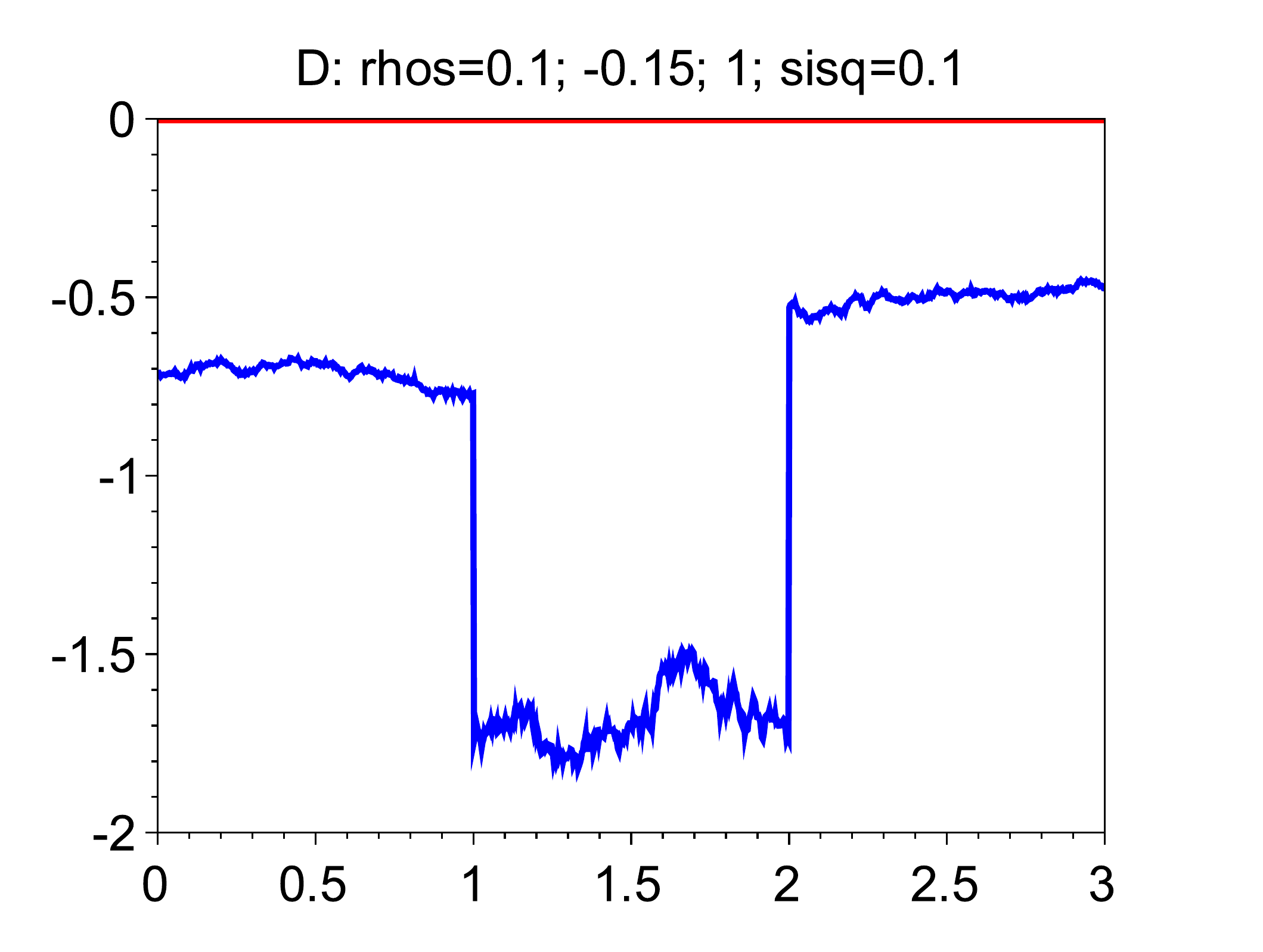}

	\caption{Top row: $\beta$, a path of $X^*$, and the corresponding path of $D^*$ for $\sigma=\sqrt{0.1}$ and $\rho^{(2)}=-0.05$. 
	Middle row: $\beta$, a path of $X^*$, and the corresponding path of $D^*$ for $\sigma=\sqrt{0.1}$ and $\rho^{(2)}=-0.07$. 
	Bottom row: $\beta$, a path of $X^*$, and the corresponding path of $D^*$ for $\sigma=\sqrt{0.1}$ and $\rho^{(2)}=-0.15$. 
	}
	\label{fig:exoverjumpingzerowithsigma}
\end{figure}

\subsection{A case study with premature closure over a time interval}\label{sec:premclosure}

In \Cref{ex:prematureclosureonlyattimepoint} the optimal strategy entails to close the position at a certain point in time and reopen it immediately. 
On the other hand, in the case $\rho \equiv 0$, it is optimal to close the position immediately and not to reenter trading (cf. \cite[Proposition 3.7]{ackermann2020cadlag}).
In the same way we can show that if, say, $\rho=0$ on $(T_1,T)$, for some $T_1\in(0,T)$, then the optimal strategy $X^*$ satisfies $X^*_.=0$ on $[T_1,T]$ (and it can involve non-trivial trading on $[0,T_1]$ depending on behaviour of the model parameters on $(0,T_1)$).
Keeping the position closed during a time interval and reopening again is
more tricky, but
also possible, as we show next. For an illustration, we refer to \Cref{fig:premclosure}.

\begin{figure}[!htb]
	\centering
	\setkeys{Gin}{width=0.48\linewidth}
	\includegraphics{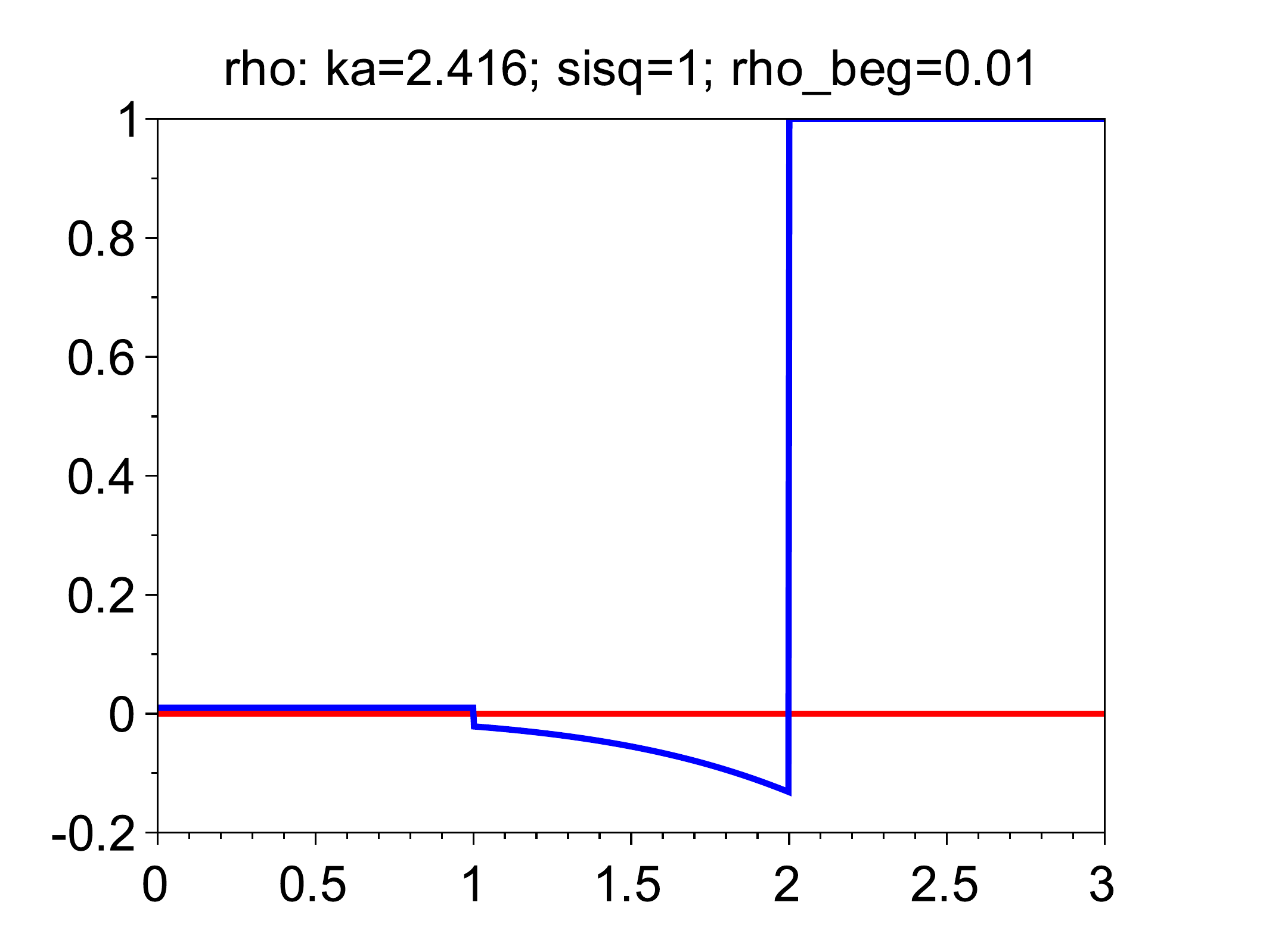}\,	\includegraphics{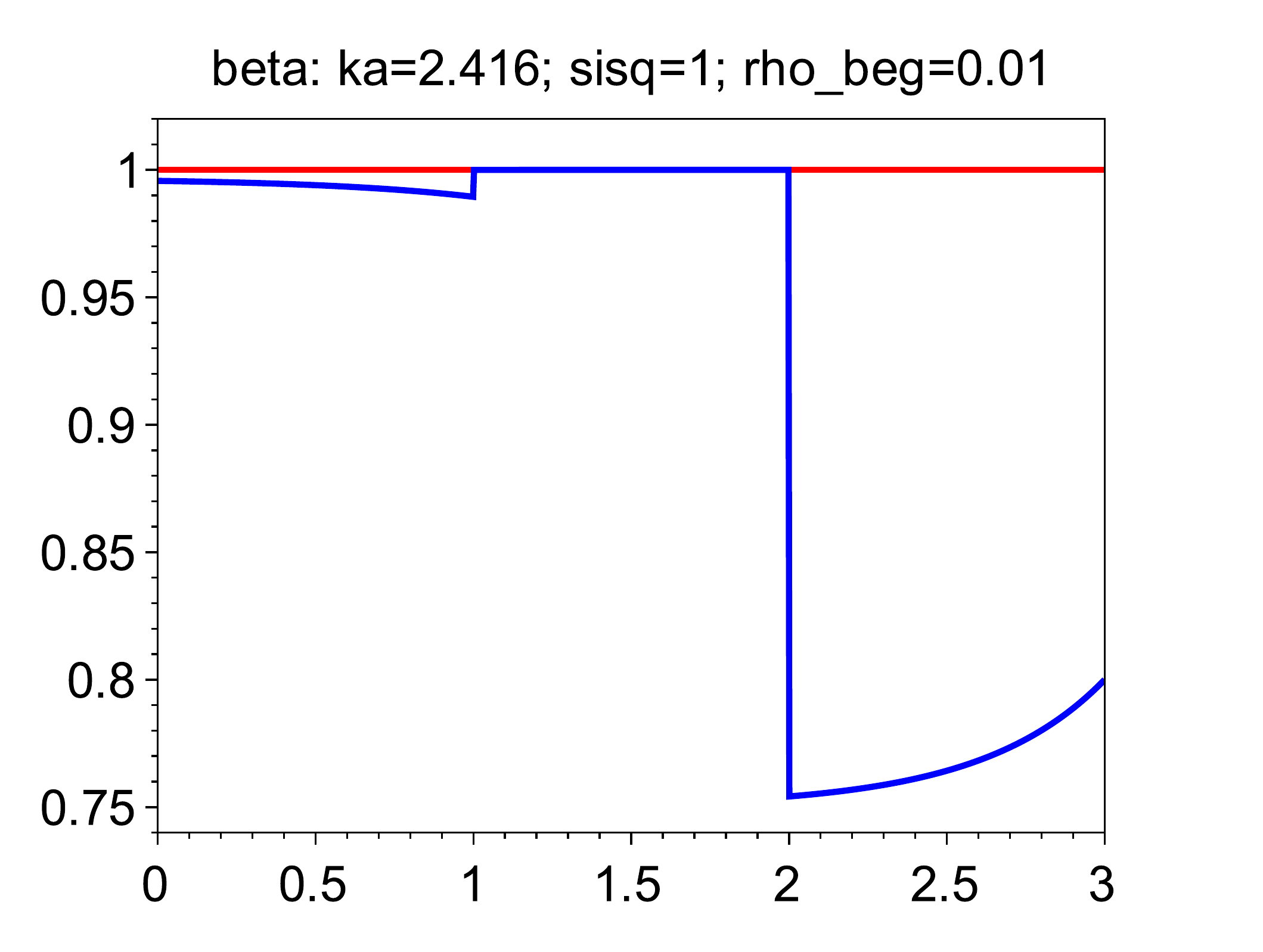} 
	
	\includegraphics{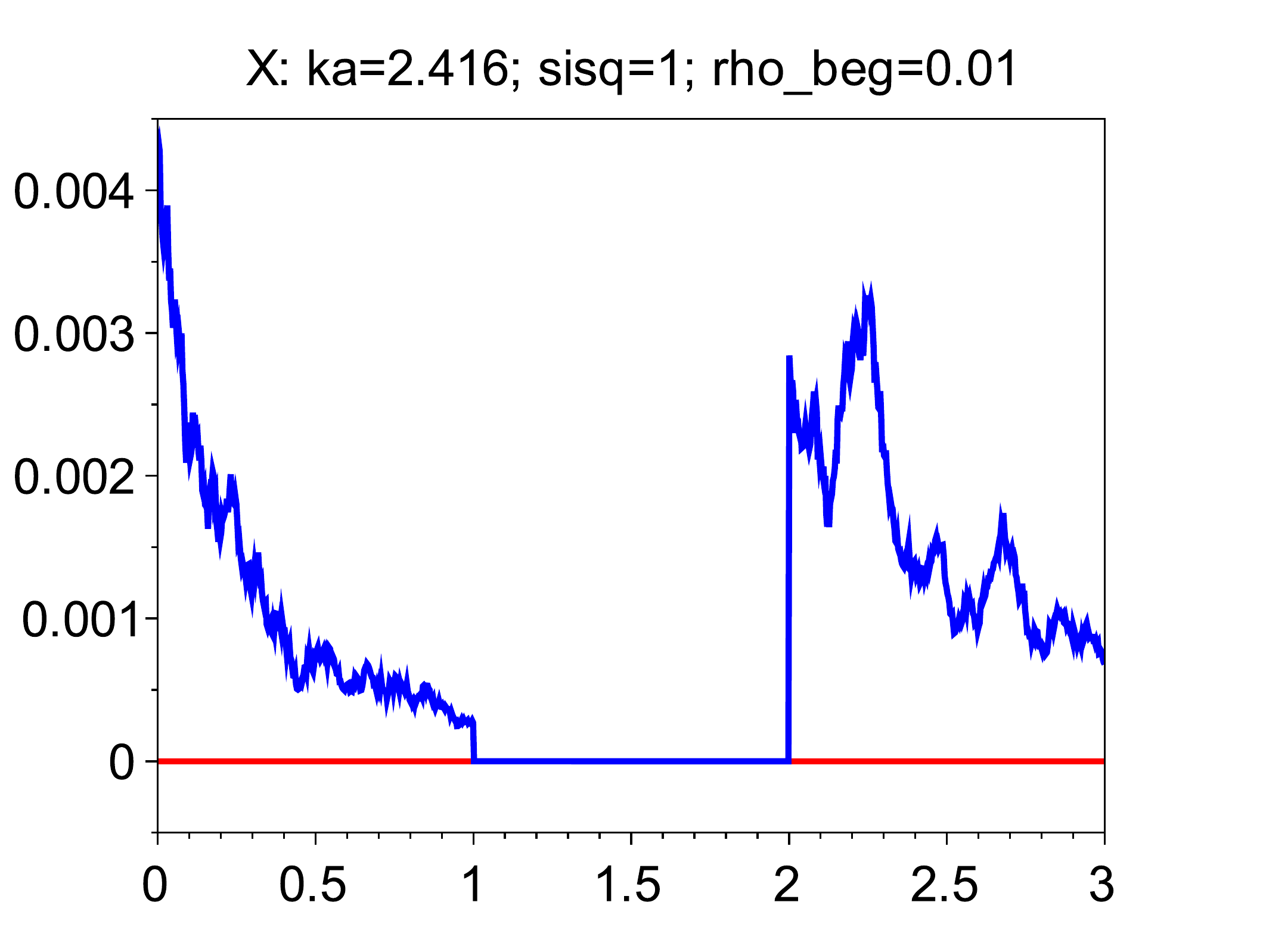}\,	\includegraphics{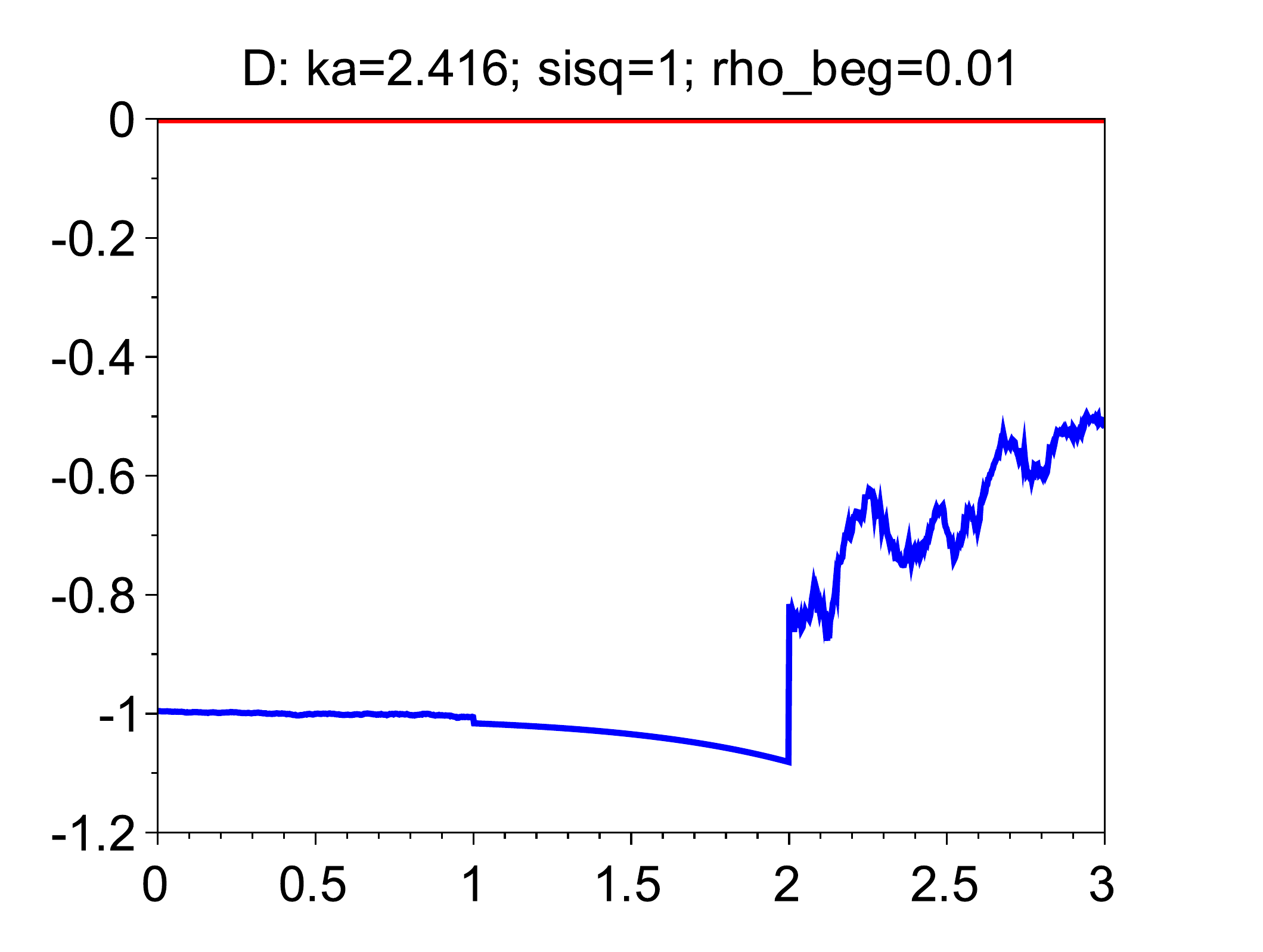}
	
	\caption{The resilience $\rho$, $\beta$, a path of the optimal strategy $X^*$, and the corresponding path of the deviation $D^*$ in the setting where $\sigma$ and $\mu=\sigma^2 +2$ are deterministic constants and $\rho$ is defined as in \eqref{eq:defrhoforpremclosure}. 
	The specific parameter values are $x=1$, $d=0$, $\gamma_0=1$, $\sigma=1$, $T=3$, $T_1=1$, $T_2=2$, $\rho^{(1)}=0.01$, $\rho^{(3)}=1$, and $\kappa=2.416$. 
	Observe that $\beta=1$ and $X^*=0$ between $t=1$ and $t=2$.}
	\label{fig:premclosure}
\end{figure}

	Let $T_1,T_2 \in (0,T)$ such that $T_1<T_2$. Suppose that $\sigma^2>0$ is a deterministic constant and that $\mu=\sigma^2 + 2$. 
	For deterministic $\rho^{(1)}>-1$, $\rho^{(3)}>0$, and $\kappa >0$ let 
	\begin{equation}\label{eq:defrhoforpremclosure}
		\rho_t = \begin{cases}
			\rho^{(1)} , & t \in [0,T_1),\\
			\left( \kappa e^{2(t-T)} + 1 \right)^{-1/2} - 1, & t \in [T_1,T_2),\\
			\rho^{(3)}, & t \in [T_2,T].
		\end{cases}
	\end{equation}
	Note that \eqref{eq:Cgeeps} and \eqref{eq:Cbdd} are satisfied. 
	Let $Y$ be the unique solution of the ODE (cf. \eqref{eq:BSDEforBM} in the current setting)
	\begin{equation}\label{eq:odeforYexconst}
		d Y_t = \frac{(\rho_t + \sigma^2 + 2)^2Y_t^2}{\sigma^2 Y_t + \rho_t +1} dt - (\sigma^2 +2) Y_t dt, \quad t \in [0,T], \quad Y_T=\frac12.
	\end{equation}
	We have \eqref{eq:exbeta} with   
	\begin{equation*}
		\beta_t = \wt\beta_t = \frac{(\rho_t + \sigma^2 + 2)Y_t}{\sigma^2 Y_t + \rho_t + 1}, \quad t\in[0,T].
	\end{equation*}
	This implies that 
	\begin{equation*}
		\{t\in[0,T]\colon \beta_t=1\} = \Big\{ t \in [0,T] \colon Y_t = \frac{\rho_t +1}{\rho_t +2} \Big\}.
	\end{equation*}
	In the sequel we establish that if $\kappa$ is chosen such that 
	$\lim_{t\uparrow T_2} \frac{\rho_t+1}{\rho_t+2}=Y_{T_2}$, then $\frac{\rho+1}{\rho+2}=Y$ on $(T_1,T_2)$. 	
To this end, suppose\footnote{Observe that to determine $Y_{T_2}$ it suffices to consider $\rho$ only on $[T_2,T]$.
In particular, $Y_{T_2}$ does not depend on the choice of $\kappa$.
Moreover, as $\rho^{(3)}\ne0$, we have $Y_{T_2}\in(0,1/2)$
(via a straightforward comparison argument for \eqref{eq:odeforYexconst}).
Therefore, we can set $\kappa=e^{2(T-T_2)}(1-2Y_{T_2}) Y_{T_2}^{-2}>0$.
It follows for this $\kappa$ that $\lim_{t\uparrow T_2} \frac{\rho_t+1}{\rho_t+2}=Y_{T_2}$.} 
	that $\lim_{t\uparrow T_2} \frac{\rho_t+1}{\rho_t+2}=Y_{T_2}$ 
	and define $\wt Y=\frac{\rho+1}{\rho+2}$ on $(T_1,T_2)$. 
	We show that $\wt Y$ is a solution of \eqref{eq:odeforYexconst} on $(T_1,T_2)$. 
	It holds for all $t\in (T_1,T_2)$ that
	\begin{equation*}
		\begin{split}
			\frac{d\wt Y_t}{dt} & = \frac{1}{(\rho_t + 2)^2} \frac{d\rho_t}{dt} 
			= -\kappa e^{2(t-T)} \frac{(\rho_t +1)^3}{(\rho_t + 2)^2} .
		\end{split}
	\end{equation*}
	On the other hand, we obtain for all $t \in (T_1,T_2)$ that 
	\begin{equation*}
		\begin{split}
			\frac{(\rho_t + \sigma^2+2)^2\wt Y_t^2}{\sigma^2 \wt Y_t + \rho_t + 1} - (\sigma^2 + 2) \wt Y_t  
			& = \left( \frac{(\rho_t + \sigma^2 + 2)^2 (\rho_t +1)}{\sigma^2 (\rho_t +1) + (\rho_t +1)(\rho_t + 2) } - (\sigma^2 +2) \right) \frac{\rho_t +1}{\rho_t + 2} \\
			& = \rho_t \frac{\rho_t +1}{\rho_t + 2}.
		\end{split}
	\end{equation*}
	In order to show that 
	\begin{equation}\label{eq:tildeYequ}
		-\kappa e^{2(t-T)} \frac{(\rho_t +1)^3}{(\rho_t + 2)^2} = \rho_t \frac{\rho_t +1}{\rho_t + 2}, \quad t \in (T_1,T_2),
	\end{equation}
	note first that this is equivalent to 
	\begin{equation*}
		-\kappa e^{2(t-T)} (\rho_t +1)^2 = \rho_t (\rho_t + 2), \quad t \in (T_1,T_2).
	\end{equation*}	
	Denoting $a_t=\kappa e^{2(t-T)}$, $t \in (T_1,T_2)$, and using $\rho_t+1=(a_t+1)^{-\frac12}$, $t\in (T_1,T_2)$, we can rewrite this as 
	\begin{equation*}
		-a_t (a_t+1)^{-1} = \left((a_t+1)^{-\frac12} - 1\right) \left( (a_t+1)^{-\frac12} + 1 \right), \quad t \in (T_1,T_2).
	\end{equation*} 
	The right hand side equals $(a_t+1)^{-1}-1$, $t \in (T_1,T_2)$. 
	We thus obtain the equivalent equation 
	\begin{equation*}
		- a_t = 1 - (a_t+1), \quad t \in (T_1,T_2),
	\end{equation*}
	which clearly holds true. 
	This proves \eqref{eq:tildeYequ}. 
	Thus, by uniqueness of the solution of \eqref{eq:odeforYexconst} and $\lim_{t\uparrow T_2} \frac{\rho_t+1}{\rho_t+2}=Y_{T_2}$, we have  $Y=\frac{\rho+1}{\rho+2}$ on $(T_1,T_2)$. 
	This implies that $\beta=1$ on $(T_1,T_2)$. 
It follows that for all $x,d \in \R$,
almost all paths of
the optimal strategy $X^*$ (cf.\ \eqref{eq:optstrat}) equal $0$ on $[T_1,T_2)$.
Finally, observe that if $x,d \in\R$ with $x\neq \frac{d}{\gamma_0}$, then
almost all paths of
$X^*$ are nonzero everywhere on $[T_2,T)$
because, on $[T_2,T)$, we have
$Y\le\frac12<\frac{\rho^{(3)}+1}{\rho^{(3)}+2}$, as $\rho^{(3)}>0$,
i.e., $Y=\frac{\rho+1}{\rho+2}$ holds nowhere on $[T_2,T)$.

\bigskip
\textbf{Acknowledgement:}
We thank two anonymous referees for suggestions that helped improve the manuscript.

\bibliographystyle{abbrv}
\bibliography{literature}

\end{document}